%% file: main.tex
\documentclass{article}

     \usepackage[final,nonatbib]{neurips_2020}

\usepackage[utf8]{inputenc} 
\usepackage[T1]{fontenc}    
\usepackage{hyperref}       
\usepackage{url}            
\usepackage{booktabs}       
\usepackage{amsfonts}       
\usepackage{nicefrac}       
\usepackage{microtype}      
\usepackage{xcolor}

\usepackage{algorithm}
\usepackage[noend]{algpseudocode}
\usepackage{amsthm,amssymb,amsmath}
\usepackage{graphicx}
\usepackage{caption}
\usepackage{subcaption}
\usepackage{mathtools}
\newtheorem{theorem}{Theorem}
\newtheorem*{theorem-non}{Theorem}
\newtheorem{lemma}{Lemma}

\DeclareCaptionSubType*{figure}

\DeclareMathOperator*{\argmax}{argmax}

\title{Combining Deep Reinforcement Learning and Search for Imperfect-Information Games}

\author{
Noam Brown\thanks{Equal contribution} \enskip Anton Bakhtin\footnotemark[1] \enskip Adam Lerer  \enskip Qucheng Gong \\
 Facebook AI Research  \\
  \texttt{\{noambrown,yolo,alerer,qucheng\}@fb.com}  \\
}

\begin{document}

\maketitle

\input{core/0_abstract.tex}
\input{core/1_introduction.tex}
\input{core/2_related_work.tex}
\input{core/3_background.tex}

\input{core/4_value.tex}
\input{core/5_method.tex}
\input{core/6_experiments.tex}
\input{core/7_results.tex}
\input{core/8_conclusions.tex}

\input{core/9_impact}

\bibliographystyle{plain}
\bibliography{references}

\newpage

\appendix

\input{core/A0_contributions}
\input{core/A1_code}
\input{core/A0_rules}
\input{core/A5_domain_knowledge}
\input{core/A7_reprodicubility}
\input{core/A8.1_dvalue_proof}
\input{core/A8.2_subgame_solving_proofs}
\input{core/A1_fp}
\input{core/A2_decomp_proof}
\input{core/A9_warm}


\end{document}

%% file: core/0_abstract.tex
\vspace{-0.05in}
\begin{abstract}
\vspace{-0.1in}
The combination of deep reinforcement learning and search at both training and test time is a powerful paradigm that has led to a number of successes in single-agent settings and perfect-information games, best exemplified by AlphaZero. However, prior algorithms of this form cannot cope with imperfect-information games. This paper presents ReBeL, a general framework for self-play reinforcement learning and search that provably converges to a Nash equilibrium in any two-player zero-sum game. In the simpler setting of perfect-information games, ReBeL reduces to an algorithm similar to AlphaZero. Results in two different imperfect-information games show ReBeL converges to an approximate Nash equilibrium. We also show ReBeL achieves superhuman performance in heads-up no-limit Texas hold'em poker, while using far less domain knowledge than any prior poker AI.
\end{abstract}

%% file: core/1_introduction.tex
\vspace{-0.05in}
\section{Introduction}
\label{sec:intro}
\vspace{-0.1in}
Combining reinforcement learning with search at both training and test time \textbf{(RL+Search)} has led to a number of major successes in AI in recent years. For example, the AlphaZero algorithm achieves state-of-the-art performance in the perfect-information games of Go, chess, and shogi~\cite{silver2018general}.

However, prior RL+Search algorithms do not work in imperfect-information games because they make a number of assumptions that no longer hold in these settings. An example of this is illustrated in Figure~\ref{fig:rps1}, which shows a modified form of Rock-Paper-Scissors in which the winner receives two points (and the loser loses two points) when either player chooses Scissors~\cite{brown2018depth}. The figure shows the game in a sequential form in which player~2 acts after player~1 but does not observe player~1's action.

\begin{figure}[!h]
	\vspace{-0.15in}
	\centering
	\begin{subfigure}[t!]{.48\textwidth}
		\centering
		\includegraphics[width=65mm]{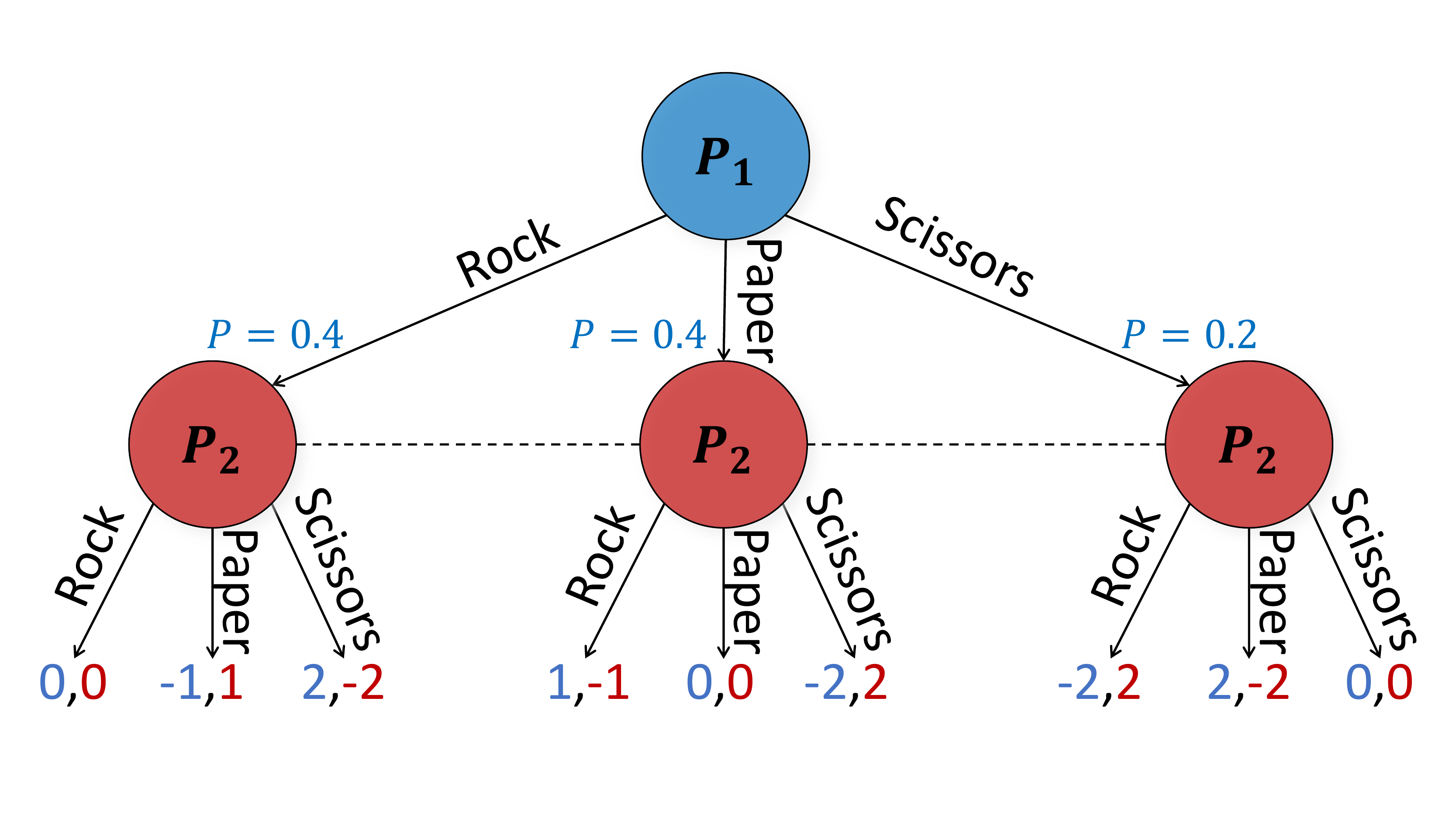}
		\vspace{-0.2in}
		\caption{Variant of Rock-Paper-Scissors in which the optimal player~1 policy is (R=0.4, P=0.4, S=0.2). Terminal values are color-coded. The dotted lines mean player~2 does not know which node they are in.}
		\label{fig:rps1}
	\end{subfigure}%
	\quad
	\begin{subfigure}[t!]{.48\textwidth}
		\centering
		\includegraphics[width=65mm]{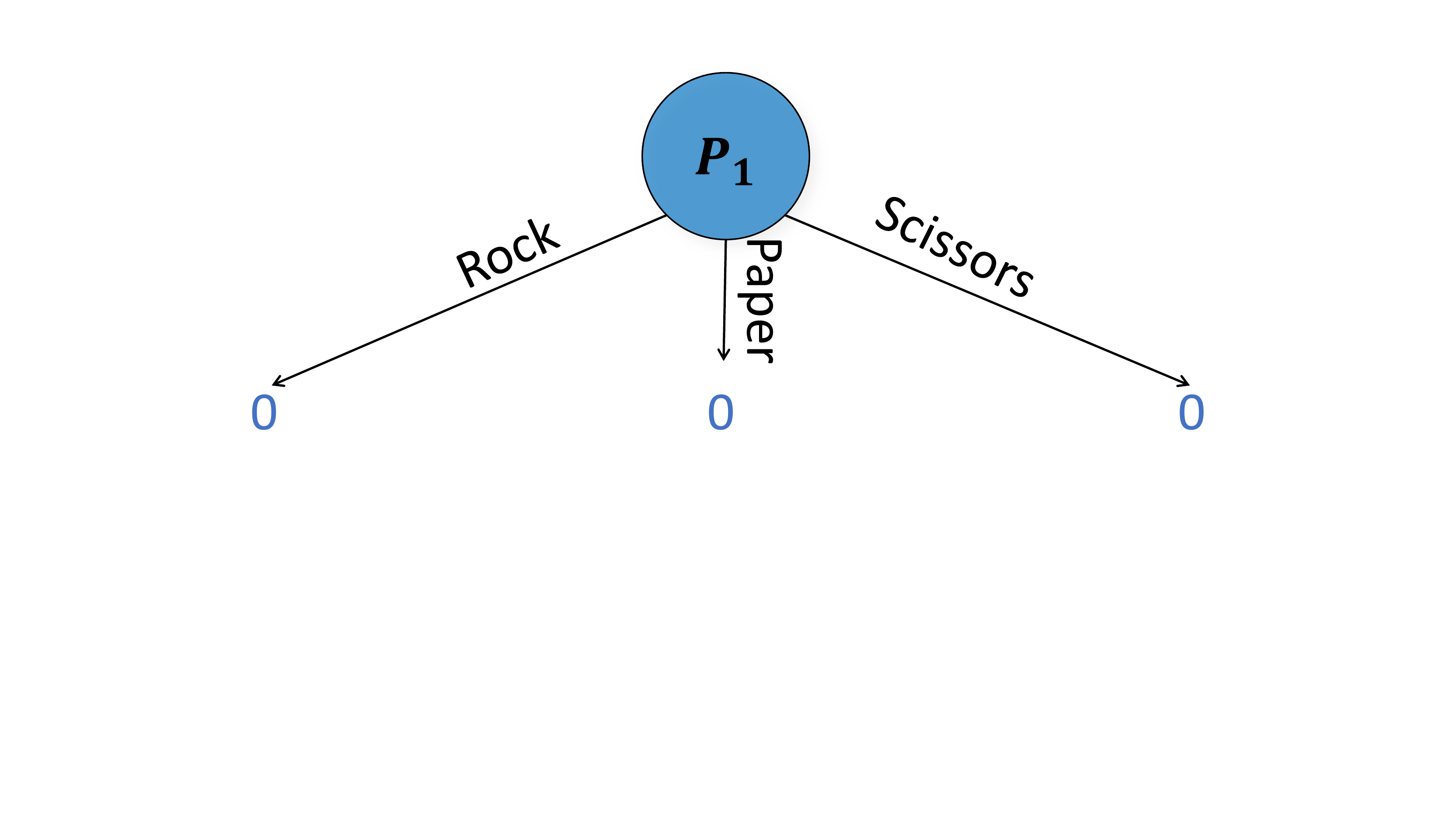}
		\vspace{-0.2in}
		\caption{The player~1 subgame when using perfect-information one-ply search. Leaf values are determined by the full-game equilibrium. There is insufficient information for finding (R=0.4, P=0.4, S=0.2).}
		\label{fig:rps2}
	\end{subfigure}
\end{figure}

The optimal policy for both players in this modified version of the game is to choose Rock and Paper with 40\% probability, and Scissors with 20\%. In that case, each action results in an expected value of zero. However, as shown in Figure~\ref{fig:rps2}, if player~1 were to conduct one-ply lookahead search as is done in perfect-information games (in which the equilibrium value of a state is substituted at a leaf node), then there would not be enough information for player~1 to arrive at this optimal policy.

This illustrates a critical challenge of imperfect-information games: unlike perfect-information games and single-agent settings, the value of an action may depend on the probability it is chosen.
Thus, a state defined only by the sequence of actions and observations does not have a unique value and therefore
existing RL+Search algorithms such as AlphaZero are not sound in imperfect-information games.
Recent AI breakthroughs in imperfect-information games have highlighted the importance of search at test time~\cite{moravvcik2017deepstack,brown2017superhuman,brown2019superhuman,lerer2020improving}, but combining RL and search during training in imperfect-information games has been an open problem.

This paper introduces ReBeL (Recursive Belief-based Learning), a general RL+Search framework that converges to a Nash equilibrium in two-player zero-sum games. ReBeL builds on prior work in which the notion of ``state'' is expanded to include the probabilistic belief distribution of all agents about what state they may be in, based on common knowledge observations and policies for all agents. Our algorithm trains a value network and a policy network for these expanded states through self-play reinforcement learning. Additionally, the algorithm uses the value and policy network for search during self play.

ReBeL provably converges to a Nash equilibrium in all two-player zero-sum games. In perfect-information games, ReBeL simplifies to an algorithm similar to AlphaZero, with the major difference being in the type of search algorithm used.
Experimental results show that ReBeL is effective in large-scale games and defeats a top human professional with statistical significance in the benchmark game of heads-up no-limit Texas hold'em poker while using far less expert domain knowledge than any previous poker AI. We also show that ReBeL approximates a Nash equilibrium in Liar's Dice, another benchmark imperfect-information game, and open source our implementation of it.\footnote{\url{https://github.com/facebookresearch/rebel}}

%% file: core/2_related_work.tex
\vspace{-0.05in}
\section{Related Work}
\vspace{-0.1in}

At a high level, ReBeL resembles past RL+Search algorithms used in perfect-information games~\cite{tesauro1994td,silver2017mastering,anthony2017thinking,silver2018general,schrittwieser2019mastering}. These algorithms train a value network through self play. During training, a search algorithm is used in which the values of leaf nodes are determined via the value function. Additionally, a policy network may be used to guide search. These forms of RL+Search have been critical to achieving superhuman performance in benchmark perfect-information games. For example, so far no AI agent has achieved superhuman performance in Go without using search at both training and test time. However, these RL+Search algorithms are not theoretically sound in imperfect-information games and have not been shown to be successful in such settings.

A critical element of our imperfect-information RL+Search framework is to use an expanded notion of ``state'', which we refer to as a \textbf{public belief state (PBS)}. PBSs are defined by a common-knowledge belief distribution over states, determined by the public observations shared by all agents and the policies of all agents. PBSs can be viewed as a multi-agent generalization of belief states used in partially observable Markov decision processes (POMDPs)~\cite{kaelbling1998planning}. The concept of PBSs originated in work on decentralized multi-agent POMDPs~\cite{nayyar2013decentralized,oliehoek2013sufficient,dibangoye2016optimally} and has been widely used since then in imperfect-information games more broadly~\cite{moravvcik2017deepstack,foerster2019bayesian,serrino2019finding,horak2019solving}.

ReBeL builds upon the idea of using a PBS value function during search, which was previously used in the poker AI DeepStack~\cite{moravvcik2017deepstack}. However, DeepStack's value function was trained not through self-play RL, but rather by generating random PBSs, including random probability distributions, and estimating their values using search. This would be like learning a value function for Go by randomly placing stones on the board.
This is not an efficient way of learning a value function because the vast majority of randomly generated situations would not be relevant in actual play. DeepStack coped with this by using handcrafted features to reduce the dimensionality of the public belief state space, by sampling PBSs from a distribution based on expert domain knowledge, and by using domain-specific abstractions to circumvent the need for a value network when close to the end of the game.

An alternative approach for depth-limited search in imperfect-information games that does not use a value function for PBSs was used in the Pluribus poker AI to defeat elite humans in multiplayer poker~\cite{brown2018depth,brown2019superhuman}. This approach trains a population of ``blueprint'' policies without using search. At test time, the approach conducts depth-limited search by allowing each agent to choose a blueprint policy from the population at leaf nodes. The value of the leaf node is the expected value of each agent playing their chosen blueprint policy against all the other agents' choice for the rest of the game.
While this approach has been successful in poker, it does not use search during training and therefore requires strong blueprint policies to be computed without search. Also, the computational cost of the search algorithm grows linearly with the number of blueprint policies.

%% file: core/3_background.tex
\vspace{-0.05in}
\section{Notation and Background}
\label{sec:background}
\vspace{-0.1in}
We assume that the rules of the game and the agents' policies (including search algorithms) are \textbf{common knowledge}~\cite{aumann1976agreeing}.\footnote{This is a common assumption in game theory. One argument for it is that in repeated play an adversary would eventually determine an agent's policy.} That is, they are known by all agents, all agents know they are known by all agents, etc. However, the outcome of stochastic algorithms (i.e., the random seeds) are not known.
In Section~\ref{sec:test} we show how to remove the assumption that we know another player's policy.

Our notation is based on that of factored observation games~\cite{kovavrik2019rethinking} which is a modification of partially observable stochastic games~\cite{hansen2004dynamic} that distinguishes between private and public observations.
We consider a game with $\mathcal{N} = \{1,2,...,N\}$ agents.

A \textbf{world state} $w \in \mathcal{W}$ is a state in the game. $\mathcal{A} = \mathcal{A}_1 \times \mathcal{A}_2 \times ... \times \mathcal{A}_N$ is the space of joint actions. $\mathcal{A}_i(w)$ denotes the legal actions for agent~$i$ at $w$ and $a = (a_1, a_2, ..., a_N) \in \mathcal{A}$ denotes a joint action. After a joint action~$a$ is chosen, a transition function $\mathcal{T}$ determines the next world state $w'$ drawn from the probability distribution $\mathcal{T}(w,a) \in \Delta \mathcal{W}$. After joint action~$a$, agent~$i$ receives a reward $\mathcal{R}_i(w,a)$.

Upon transition from world state $w$ to $w'$ via joint action $a$, agent~$i$ receives a \textbf{private observation} from a function $\mathcal{O}_{\text{priv}(i)}(w,a,w')$. Additionally, all agents receive a \textbf{public observation} from a function $\mathcal{O}_{\text{pub}}(w,a,w')$. Public observations may include observations of publicly taken actions by agents. For example, in many recreational games, including poker, all betting actions are public.

A \textbf{history} (also called a trajectory) is a finite sequence of legal actions and world states, denoted $h = (w^0, a^0, w^1, a^1, ..., w^t)$. An \textbf{infostate} (also called an \textbf{action-observation history (AOH)}) for agent~$i$ is a sequence of an agent's observations and actions $s_i = (O_i^0, a_i^0, O_i^1, a_i^1, ..., O_i^t)$ where $O_i^k = \big(\mathcal{O}_{\text{priv}(i)}(w^{k-1},a^{k-1},w^k), \mathcal{O}_{\text{pub}}(w^{k-1},a^{k-1},w^k)\big)$.
The unique infostate corresponding to a history $h$ for agent~$i$ is denoted $s_i(h)$. The set of histories that correspond to $s_i$ is denoted $\mathcal{H}(s_i)$.

A \textbf{public state} is a sequence $s_{\text{pub}} = (O_{\text{pub}}^0, O_{\text{pub}}^1, ..., O_{\text{pub}}^t)$ of public observations. The unique public state corresponding to a history $h$ and an infostate $s_i$ is denoted $s_{\text{pub}}(h)$ and $s_{\text{pub}}(s_i)$, respectively. The set of histories that match the sequence of public observation of $s_\text{pub}$ is denoted $\mathcal{H}(s_\text{pub})$.

For example, consider a game where two players roll two six-sided dice each. One die of each player is publicly visible; the other die is only observed by the player who rolled it. Suppose player~$1$ rolls a $3$ and a $4$ (with $3$ being the hidden die), and player~$2$ rolls a $5$ and a $6$ (with $5$ being the hidden die). The history (and world state) is $h = \big((3,4),(5,6)\big)$. The set of histories corresponding to player~$2$'s infostate is $\mathcal{H}(s_2) = \big\{\big((x,4),(5,6)\big) \mid x \in \{1,2,3,4,5,6\}\big\}$, so $|\mathcal{H}(s_2)| = 6$. The set of histories corresponding to $s_{\text{pub}}$ is $\mathcal{H}(s_{\text{pub}}) = \big\{\big((x,4),(y,6)\big) \mid x,y \in \{1,2,3,4,5,6\}\big\}$, so $|\mathcal{H}(s_{\text{pub}})| = 36$ .

Public states provide an easy way to reason about common knowledge in a game. All agents observe the same public sequence $s_\text{pub}$, and therefore it is common knowledge among all agents that the true history is some $h \in \mathcal{H}(s_\text{pub})$.\footnote{As explained in~\cite{kovavrik2019rethinking}, it may be possible for agents to infer common knowledge beyond just public observations.
However, doing this additional reasoning is inefficient both theoretically and practically.}

An agent's \textbf{policy} $\pi_i$ is a function mapping from an infostate to a probability distribution over actions. A \textbf{policy profile} $\pi$ is a tuple of policies $(\pi_1, \pi_2, ..., \pi_N)$.
The expected sum of future rewards (also called the \textbf{expected value (EV)}) for agent~$i$ in history $h$ when all agents play policy profile $\pi$ is denoted $v_i^{\pi}(h)$. The EV for the entire game is denoted $v_i(\pi)$. A \textbf{Nash equilibrium} is a policy profile such that no agent can achieve a higher EV by switching to a different policy~\cite{nash1951non}. Formally, $\pi^*$ is a Nash equilibrium if for every agent~$i$, $v_i(\pi^*) = \max_{\pi_i}v_i(\pi_i, \pi^*_{-i})$ where $\pi_{-i}$ denotes the policy of all agents other than $i$. A \textbf{Nash equilibrium policy} is a policy $\pi^*_i$ that is part of some Nash equilibrium $\pi^*$.

A \textbf{subgame} is defined by a root history $h$ in a perfect-information game and all histories that can be reached going forward. In other words, it is identical to the original game except it starts at $h$. A \textbf{depth-limited subgame} is a subgame that extends only for a limited number of actions into the future. Histories at the bottom of a depth-limited subgame (i.e., histories that have no legal actions in the depth-limited subgame) but that have at least one legal action in the full game are called \textbf{leaf nodes}. In this paper, we assume for simplicity that search is performed over fixed-size depth-limited subgame (as opposed to Monte Carlo Tree Search, which grows the subgame over time~\cite{gelly2007combining}).

A game is \textbf{two-player zero-sum (2p0s)} if there are exactly two players and $\mathcal{R}_1(w,a) = -\mathcal{R}_2(w,a)$ for every world state~$w$ and action~$a$.
In 2p0s perfect-information games, there always exists a Nash equilibrium that depends only on the current world state~$w$ rather than the entire history~$h$. Thus, in 2p0s perfect-information games a policy can be defined for world states and a subgame can be defined as rooted at a world state. Additionally, in 2p0s perfect-information games every world state $w$ has a unique value $v_i(w)$ for each agent~$i$, where $v_1(w) = -v_2(w)$, defined by both agents playing a Nash equilibrium in any subgame rooted at that world state. Our theoretical and empirical results are limited to 2p0s games, though related techniques have been empirically successful in some settings with more players~\cite{brown2019superhuman}.
A typical goal for RL in 2p0s perfect-information games is to learn $v_i$.
With that value function, an agent can compute its optimal next move by solving a depth-limited subgame that is rooted at its current world state and where the value of every leaf node $z$ is set to $v_i(z)$~\cite{shannon1950programming,samuel1959some}.

%% file: core/4_value.tex
\vspace{-0.05in}
\section{From World States to Public Belief States}
\vspace{-0.1in}
\label{sec:pbs}

In this section we describe a mechanism for converting any imperfect-information game into a continuous state (and action) space perfect-information game where the state description contains the probabilistic belief distribution of all agents. In this way, techniques that have been applied to perfect-information games can also be applied to imperfect-information games (with some modifications).

For intuition, consider a game in which one of 52 cards is privately dealt to each player. On each turn, a player chooses between three actions: fold, call, and raise. Eventually the game ends and players receive a reward. Now consider a modification of this game in which the players cannot see their private cards; instead, their cards are seen by a ``referee''. On a player's turn, they announce the probability they would take each action with each possible private card. The referee then samples an action on the player's behalf from the announced probability distribution for the player's true private card. When this game starts, each player's belief distribution about their private card is uniform random. However, after each action by the referee, players can update their belief distribution about which card they are holding via Bayes' Rule. Likewise, players can update their belief distribution about the \emph{opponent's} private card through the same operation. Thus, the probability that each player is holding each private card is common knowledge among all players at all times in this game.

A critical insight is that \textit{these two games are strategically identical}, but the latter contains no private information and is instead a continuous state (and action) space perfect-information game. While players do not announce their action probabilities for each possible card in the first game, we assume (as stated in Section~\ref{sec:background}) that all players' policies are common knowledge, and therefore the probability that a player would choose each action for each possible card is indeed known by all players.
Of course, at test time (e.g., when our agent actually plays against a human opponent) the opponent does not actually announce their entire policy and therefore our agent does not know the true probability distribution over opponent cards. We later address this problem in Section~\ref{sec:test}.

We refer to the first game as the \textbf{discrete representation} and the second game as the \textbf{belief representation}. In the example above, a history in the belief representation, which we refer to as a \textbf{public belief state (PBS)}, is described by the sequence of public observations and 104 probabilities (the probability that each player holds each of the 52 possible private card); an ``action'' is described by 156 probabilities (one per discrete action per private card). In general terms, a PBS is described by a joint probability distribution over the agents' possible infostates~\cite{nayyar2013decentralized,oliehoek2013sufficient,dibangoye2016optimally}.\footnote{One could alternatively define a PBS as a probability distribution over histories in $\mathcal{H}(s_\text{pub})$ for public state $s_\text{pub}$. However, it is proven that any PBS that can arise in play can always be described by a joint probability distribution over the agents' possible infostates~\cite{oliehoek2013sufficient,seitz2019value}, so we use this latter definition for simplicity.} Formally, let $S_i(s_\text{pub})$ be the set of infostates that player $i$ may be in given a public state $s_\text{pub}$ and let $\Delta S_1(s_\text{pub})$ denote a probability distribution over the elements of $S_1(s_\text{pub})$. Then PBS $\beta = (\Delta S_1(s_\text{pub}), ..., \Delta S_N(s_\text{pub}))$.\footnote{Frequently, a PBS can be compactly summarized by discarding parts of the history that are no longer relevant. For example, in poker we do not need to track the entire history of actions, but just the amount of money each player has in the pot, the public board cards, and whether there were any bets in the current round.} In perfect-information games, the discrete representation and belief representation are identical.

Since a PBS is a history of the perfect-information belief-representation game, a subgame can be rooted at a PBS.\footnote{Past work defines a subgame to be rooted at a public state~\cite{burch2014solving,brown2015simultaneous,moravcik2016refining,moravvcik2017deepstack,brown2017safe,kovavrik2019problems,sustr2019monte,seitz2019value}. However, imperfect-information subgames rooted at a public state do not have well-defined values.}
The discrete-representation interpretation of such a subgame is that at the start of the subgame a history is sampled from the joint probability distribution of the PBS, and then the game proceeds as it would in the original game. The value for agent~$i$ of PBS~$\beta$ when all players play policy profile $\pi$ is $V_i^{\pi}(\beta) = \sum_{h \in \mathcal{H}(s_{\text{pub}}(\beta))} \left(p(h|\beta) v_i^{\pi}(h)\right)$.
Just as world states have unique values in 2p0s perfect-information games, in 2p0s games (both perfect-information and imperfect-information) every PBS $\beta$ has a unique value $V_i(\beta)$ for each agent~$i$, where $V_1(\beta) = -V_2(\beta)$, defined by both players playing a Nash equilibrium in the subgame rooted at the PBS.

Since any imperfect-information game can be viewed as a perfect-information game consisting of PBSs (i.e., the belief representation), in theory we could approximate a solution of any 2p0s imperfect-information game by running a perfect-information RL+Search algorithm on a discretization of the belief representation. However, as shown in the example above, belief representations can be very high-dimensional continuous spaces, so conducting search (i.e., approximating the optimal policy in a depth-limited subgame) as is done in perfect-information games would be intractable. Fortunately, in 2p0s games, \textit{these high-dimensional belief representations are convex optimization problems.} ReBeL leverages this fact by conducting search via an iterative gradient-ascent-like algorithm.

ReBeL's search algorithm operates on supergradients (the equivalent of subgradients but for concave functions) of the PBS value function at leaf nodes, rather than on PBS values directly. Specifically, the search algorithms require the values of \emph{infostates} for PBSs~\cite{burch2014solving,moravvcik2017deepstack}.
In a 2p0sum game, the value of infostate $s_i$ in $\beta$ assuming all other players play Nash equilibrium $\pi^*$ is the maximum value that player~$i$ could obtain for $s_i$ through any policy in the subgame rooted at $\beta$. Formally,
\begin{equation}
v_i^{\pi^*}(s_i|\beta) = \max_{\pi_i} \sum_{h \in \mathcal{H}(s_i)} p(h | s_i, \beta_{-i}) v_i^{\langle \pi_i, \pi^*_{-i} \rangle}(h)
\label{eq:infostate_value}
\end{equation}
where $p(h|s_i, \beta_{-i})$ is the probability of being in history $h$ assuming $s_i$ is reached and the joint probability distribution over infostates for players other than~$i$ is $\beta_{-i}$.
Theorem~\ref{th:cfv} proves that infostate values can be interpreted as a supergradient of the PBS value function in 2p0s games.
\begin{theorem}
\label{th:cfv}
For any PBS $\beta=(\beta_1,\beta_2)$ (for the beliefs over player 1 and 2 infostates respectively) and any policy $\pi^*$ that is a Nash equilibrium of the subgame rooted at $\beta$,
\begin{equation}
    v_1^{\pi^*}(s_1|\beta) = V_1(\beta) +\bar{g} \cdot \hat{s}_1
\end{equation} 
where $\bar{g}$ is a supergradient of an extension of $V_1(\beta)$ to unnormalized belief distributions and $\hat{s}_1$ is the unit vector in direction $s_1$.
\end{theorem}
All proofs are presented in the appendix.

Since ReBeL's search algorithm uses infostate values, so rather than learn a PBS value function ReBeL instead learns an infostate-value function $\hat{v}: \mathcal{B} \rightarrow \mathbb{R}^{|S_1|+|S_2|}$ that directly approximates for each $s_i$ the average of the sampled $v_i^{\pi^*}(s_i|\beta)$ values produced by ReBeL at $\beta$.\footnote{Unlike the PBS value $V_i(\beta)$, the infostate values may not be unique and may depend on which Nash equilibrium is played in the subgame. Nevertheless,
any linear combination of supergradients is itself a supergradient since the set of all supergradients is a convex set~\cite{rockafellar1970convex}.
}

%% file: core/5_method.tex
\vspace{-0.05in}
\section{Self Play Reinforcement Learning and Search for Public Belief States}
\vspace{-0.1in}
\label{sec:method}
In this section we describe ReBeL and prove that it approximates a Nash equilibrium in 2p0s games.
At the start of the game, a depth-limited subgame rooted at the initial PBS $\beta_r$ is generated. This subgame is solved (i.e., a Nash equilibrium is approximated) by running $T$ iterations of an iterative equilibrium-finding algorithm in the discrete representation of the game, but using the learned value network~$\hat{v}$ to approximate leaf values on every iteration.
During training, the infostate values at $\beta_r$ computed during search are added as training examples for $\hat{v}$ and (optionally) the subgame policies are added as training examples for the policy network. Next, a leaf node~$z$ is sampled and the process repeats
with the PBS at $z$ being the new subgame root. Appendix~\ref{sec:code} shows detailed pseudocode.

\begin{algorithm}
\caption{ReBeL: RL and Search for Imperfect-Information Games}

\begin{algorithmic}
\Function{SelfPlay}{$\beta_r, \theta^v, \theta^\pi, D^v,D^\pi$} \Comment{$\beta_r$ is the current PBS}
    \While{!\Call{IsTerminal}{$\beta_r$}}
    \State $G \gets $ \Call{ConstructSubgame}{$\beta_r$}
    \State $\bar{\pi}, \pi^{t_{\text{warm}}} \gets $ \Call{InitializePolicy}{$G, \theta^\pi$} \Comment{$t_{\text{warm}} = 0$ and $\pi^0$ is uniform if no warm start}
    \State $G \gets $ \Call{SetLeafValues}{$G, \bar{\pi}, \pi^{t_{\text{warm}}}, \theta^v$}
    \State $v(\beta_r) \gets $ \Call{ComputeEV}{$G,\pi^{t_\text{warm}}$}
    \State $t_{\textit{sample}} \sim \mathrm{unif}\{t_{\text{warm}}+1, T\}$ \Comment{Sample an iteration}
    \For{$t=(t_{\text{warm}}+1)..T$}
        \If{$t = t_{\textit{sample}}$}
            \State $\beta'_r \gets $ \Call{SampleLeaf}{$G,\pi^{t-1}$} \Comment{Sample one or multiple leaf PBSs}
        \EndIf
        \State $\pi^{t} \gets $ \Call{UpdatePolicy}{$G, \pi^{t-1}$}
        \State $\bar{\pi} \gets \frac{t}{t+1} \bar{\pi} + \frac{1}{t+1} \pi^{t}$
        \State $G \gets $ \Call{SetLeafValues}{$G, \bar{\pi}, \pi^{t}, \theta^v$}
        \State $v(\beta_r) \gets \frac{t}{t+1} v(\beta_r) + \frac{1}{t+1}$ \Call{ComputeEV}{$G,\pi^{t}$}
    \EndFor
    \State Add $\{\beta_r,v(\beta_r)\}$ to $D^v$ \Comment{Add to value net training data}
    \For{$\beta \in G$} \Comment{Loop over the PBS at every public state in $G$}
    \State Add $\{\beta,\bar{\pi}(\beta)\}$ to $D^\pi$ \Comment{Add to policy net training data (optional)}
    \EndFor
    $\beta_r \gets \beta'_r$
    \EndWhile
\EndFunction
\end{algorithmic}
\label{alg:rl_loop}
\end{algorithm}
\vspace{-0.1in}
\subsection{Search in a depth-limited imperfect-information subgame}
\vspace{-0.1in}
\label{sec:method_low}
In this section we describe the search algorithm ReBeL uses to solve depth-limited subgames.
We assume for simplicity that the depth of the subgame is pre-determined and fixed.
The subgame is solved in the discrete representation and the solution is then converted to the belief representation. There exist a number of iterative algorithms for solving imperfect-information games~\cite{brown1951iterative,zinkevich2008regret,hoda2010smoothing,kroer2018faster,kroer2018solving}.
We describe ReBeL assuming the \textbf{counterfactual regret minimization - decomposition (CFR-D)} algorithm is used~\cite{zinkevich2008regret,burch2014solving,moravvcik2017deepstack}. CFR is the most popular equilibrium-finding algorithm for imperfect-information games, and CFR-D is an algorithm that solves depth-limited subgames via CFR.
However, ReBeL is flexible with respect to the choice of search algorithm and in Section~\ref{sec:results} we also show experimental results for \textbf{fictitious play (FP)}~\cite{brown1951iterative}.

On each iteration $t$, CFR-D determines a policy profile $\pi^t$ in the subgame. Next, the value of every discrete representation leaf node~$z$ is set to $\hat{v}(s_i(z)| \beta_z^{\pi^t})$, where $\beta_z^{\pi^t}$ denotes the PBS at $z$ when agents play according to $\pi^t$.
This means that the value of a leaf node during search is conditional on $\pi^t$. Thus, the leaf node values change every iteration. Given $\pi^t$ and the leaf node values, each infostate in $\beta_r$ has a well-defined value. This vector of values, denoted $v^{\pi^t}(\beta_r)$, is stored.
Next, CFR-D chooses a new policy profile $\pi^{t+1}$, and the process repeats for $T$ iterations.

When using CFR-D, the \emph{average} policy profile $\bar{\pi}^T$ converges to a Nash equilibrium as $T \to \infty$, rather than the policy on the final iteration.
Therefore, after running CFR-D for $T$ iterations in the subgame rooted at PBS $\beta_{r}$, the value vector $(\sum_{t=1}^T v^{\pi^t}(\beta_r))/T$ is added to the training data for $\hat{v}(\beta_r)$.

Appendix~\ref{sec:cfrave} introduces \textbf{CFR-AVG}, a modification of CFR-D that sets the value of leaf node~$z$ to $\hat{v}(s_i(z)| \beta_z^{\bar{\pi}^t})$ rather than $\hat{v}(s_i(z)| \beta_z^{\pi^t})$, where $\bar{\pi}^t$ denotes the average policy profile up to iteration~$t$. CFR-AVG addresses some weaknesses of CFR-D.
\vspace{-0.1in}
\subsection{Self-play reinforcement learning}
\vspace{-0.1in}
\label{sec:method_high}
We now explain how ReBeL trains a PBS value network through self play. After solving a subgame rooted at PBS~$\beta_{r}$ via search (as described in Section~\ref{sec:method_low}), the value vector for the root infostates is added to the training dataset for $\hat{v}$. Next, a leaf PBS~$\beta'_r$ is sampled and a new subgame rooted at $\beta'_r$ is solved. This process repeats until the game ends.

Since the subgames are solved using an iterative algorithm, we want $\hat{v}$ to be accurate for leaf PBSs on every iteration. Therefore, a leaf node $z$ is sampled according to $\pi^t$ on a \emph{random} iteration $t \sim \mathrm{unif}\{0, T-1\}$, where $T$ is the number of iterations of the search algorithm.\footnote{For FP, we pick a random agent~$i$ and sample according to $(\pi_i^t, \bar{\pi}_{-i}^t)$ to reflect the search operation.} To ensure sufficient exploration, one agent samples random actions with probabilility $\epsilon > 0$.\footnote{The algorithm is still correct if all agents sample random actions with probability $\epsilon$, but that is less efficient because the value of a leaf node that can only be reached if both agents go off policy is irrelevant.} 
In CFR-D $\beta'_r = \beta^{\pi^t}_z$, while in CFR-AVG and FP $\beta'_r = \beta^{\bar{\pi}^t}_z$.

Theorem~\ref{th:perfect} states that, with perfect function approximation, running Algorithm~\ref{alg:rl_loop} will produce a value network whose error is bounded by $\mathcal{O}(\frac{1}{\sqrt{T}})$ for any PBS that could be encountered during play, where $T$ is the number of CFR iterations being run in subgames.

\begin{theorem}
\label{th:perfect}
Consider an idealized value approximator that returns the most recent sample of the value for sampled PBSs, and 0 otherwise.  Running Algorithm~\ref{alg:rl_loop} with $T$ iterations of CFR in each subgame will produce a value approximator that has error of at most $\frac{C}{\sqrt{T}}$ for any PBS that could be encountered during play, where $C$ is a game-dependent constant.
\end{theorem}

ReBeL as described so far trains the value network through bootstrapping. One could alternatively train the value network using rewards actually received over the course of the game when the agents do not go off-policy. There is a trade-off between bias and variance between these two approaches~\cite{schulman2015high}.

\vspace{-0.1in}
\subsection{Adding a policy network}
\vspace{-0.1in}
\label{sec:method_policy}
Algorithm~\ref{alg:rl_loop} will result in $\hat{v}$ converging correctly even if a policy network is not used. However, initializing the subgame policy via a policy network may reduce the number of iterations needed to closely approximate a Nash equilibrium. Additionally, it may improve the accuracy of the value network by allowing the value network to focus on predicting PBS values over a more narrow domain.

Algorithm~\ref{alg:rl_loop} can train a policy network $\hat{\Pi} : \mathcal{\beta} \to (\Delta \mathcal{A})^{|S_1|+|S_2|}$ by adding $\bar{\pi}^T(\beta)$ for each PBS~$\beta$ in the subgame to a training dataset each time a subgame is solved (i.e., $T$ iterations of CFR have been run in the subgame).
Appendix~\ref{sec:warm} describes a technique, based on~\cite{brown2016strategy}, for warm starting equilibrium finding given the initial policy from the policy network.

\vspace{-0.05in}
\section{Playing According to an Equilibrium at Test Time}
\vspace{-0.1in}
\label{sec:test}

This section proves that running Algorithm~\ref{alg:rl_loop} at test time with an accurately trained PBS value network will result in playing a Nash equilibrium policy in expectation even if we do not know the opponent's policy.
During self play training we assumed, as stated in Section~\ref{sec:background}, that both players' policies are common knowledge. This allows us to exactly compute the PBS we are in. However, at test time we do not know our opponent's entire policy, and therefore we do not know the PBS. This is a problem for conducting search, because search is always rooted at a PBS.
For example, consider again the game of modified Rock-Paper-Scissors illustrated in Figure~\ref{fig:rps1}. For simplicity, assume that $\hat{v}$ is perfect. Suppose that we are player~2 and player~1 has just acted. In order to now conduct search as player~$2$, our algorithm requires a root PBS. What should this PBS be?

An intuitive choice, referred to as \textbf{unsafe} search~\cite{gilpin2006competitive,ganzfried2015endgame}, is to first run CFR for $T$ iterations for player~1's first move (for some large $T$), which results in a player~1 policy such as $(R = 0.4001, P = 0.3999, S = 0.2)$.
Unsafe search passes down the beliefs resulting from that policy, and then computes our optimal policy as player~$2$. This would result in a policy of $(R = 0, P = 1, S = 0)$ for player~$2$. Clearly, this is not a Nash equilibrium. Moreover, if our opponent knew we would end up playing this policy (which we assume they would know since we assume they know the algorithm we run to generate the policy), then they could exploit us by playing $(R = 0, P = 0, S = 1)$.

This problem demonstrates the need for \textbf{safe} search, which is a search algorithm that ensures we play a Nash equilibrium policy in expectation. Importantly, it is \emph{not} necessary for the policy that the algorithm outputs to always be a Nash equilibrium. It is only necessary that the algorithm outputs a Nash equilibrium policy \emph{in expectation}. For example, in modified Rock-Paper-Scissors it is fine for an algorithm to output a policy of 100\% Rock, so long as the probability it outputs that policy is~40\%.

All past safe search approaches introduce additional constraints to the search algorithm~\cite{burch2014solving,moravcik2016refining,brown2017safe,sustr2019monte}. Those additional constraints
hurt performance in practice compared to unsafe search~\cite{burch2014solving,brown2017safe} and greatly complicate search, so they were never fully used in any competitive agent. Instead, all previous search-based imperfect-information game agents used unsafe search either partially or entirely~\cite{moravvcik2017deepstack,brown2017superhuman,brown2018depth,brown2019superhuman,serrino2019finding}.
Moreover,
using prior safe search techniques at test time
may result in the agent encountering PBSs that were not encountered during self-play training and therefore may result in poor approximations from the value and policy network.

We now prove that safe search can be achieved without any additional constraints by simply \textit{running the same algorithm at test time that we described for training}. This result applies regardless of how the value network was trained and so can be applied to prior algorithms that use PBS value functions~\cite{moravvcik2017deepstack,serrino2019finding}. Specifically, when conducting search at test time we pick a \textit{random} iteration and assume all players' policies match the policies on that iteration.
Theorem~\ref{th:safe}, the proof of which is in Section~\ref{sec:proof_safe}, states that once a value network is trained according to Theorem~\ref{th:perfect}, using Algorithm~\ref{alg:rl_loop} at test time (without off-policy exploration) will approximate a Nash equilibrium.

\begin{theorem}
\label{th:safe}
If Algorithm~\ref{alg:rl_loop} is run at test time with no off-policy exploration, a value network with error at most $\delta$ for any leaf PBS that was trained to convergence as described in Theorem~\ref{th:perfect}, and with $T$ iterations of CFR being used to solve subgames, then the algorithm plays a $(\delta C_1 + \frac{\delta C_2}{\sqrt{T}})$-Nash equilibrium, where $C_1, C_2$ are game-specific constants.
\end{theorem}

Since a random iteration is selected, we may select a very early iteration, or even the first iteration, in which the policy is extremely poor. This can be mitigated by using modern equilibrium-finding algorithms, such as Linear CFR~\cite{brown2019solving}, that assign little or no weight to the early iterations.

%% file: core/6_experiments.tex
\vspace{-0.05in}
\section{Experimental Setup}
\label{sec:training}
\vspace{-0.1in}
We measure \textbf{exploitability} of a policy~$\pi^*$, which is $\sum_{i \in \mathcal{N}}\max_{\pi}v_i(\pi, \pi^*_{-i}) / |\mathcal{N}|$. All CFR experiments use alternating-updates Linear CFR~\cite{brown2019solving}. All FP experiments use alternating-updates Linear Optimistic FP, which is a novel variant we present in Appendix~\ref{sec:fp}.

We evaluate on the benchmark imperfect-information games of heads-up no-limit Texas hold'em poker (HUNL) and Liar's Dice. The rules for both games are provided in Appendix~\ref{sec:games}. We also evaluate our techniques on turn endgame hold'em (TEH), a variant of no-limit Texas hold'em in which both players automatically check/call for the first two of the four betting rounds in the game.

In HUNL and TEH, we reduce the action space to at most nine actions using domain knowledge of typical bet sizes. However, our agent responds to any ``off-tree'' action at test time by adding the action to the subgame~\cite{brown2018depth,brown2019superhuman}. The bet sizes and stack sizes are randomized during training. For TEH we train on the full game and measure exploitability on the case of both players having \$20,000, unperturbed bet sizes, and the first four board cards being $3 \spadesuit$7$\heartsuit$T$\diamondsuit$K$\spadesuit$.
For HUNL, our agent uses far less domain knowledge than any prior competitive AI agent.
Appendix~\ref{sec:dk} discusses the poker domain knowledge we leveraged in ReBeL.

We approximate the value and policy functions using artificial neural networks. Both networks are MLPs with GeLU~\cite{hendrycks2016gaussian} activation functions and LayerNorm~\cite{ba2016layer}. Both networks are trained with Adam~\cite{kingma2014adam}. We use pointwise Huber loss as the criterion for the value function and mean squared error (MSE) over probabilities for the policy.
In preliminary experiments we found MSE for the value network and cross entropy for the policy network did worse.
See Appendix~\ref{sec:hyperparams} for the hyperparameters.

We use PyTorch~\cite{paszke2019pytorch} to train the networks. We found data generation to be the bottleneck due to the sequential nature of the FP and CFR algorithms and the evaluation of all leaf nodes on each iteration. For this reason we use a single machine for training and up to 128 machines with 8 GPUs each for data generation.

%% file: core/7_results.tex
\vspace{-0.05in}
\section{Experimental Results}
\label{sec:results}
\vspace{-0.1in}
Figure~\ref{fig:turn_cfr} shows ReBeL reaches a level of exploitability in TEH equivalent to running about 125 iterations of full-game tabular CFR. For context, top poker agents typically use between 100 and 1,000 tabular CFR iterations~\cite{bowling2015heads,moravvcik2017deepstack,brown2017superhuman,brown2018depth,brown2019superhuman}. Our self-play algorithm is key to this success; Figure~\ref{fig:turn_cfr} shows a value network trained on random PBSs fails to learn anything valuable.

\begin{figure}[!h]
	\centering
	\includegraphics[width=1\textwidth]{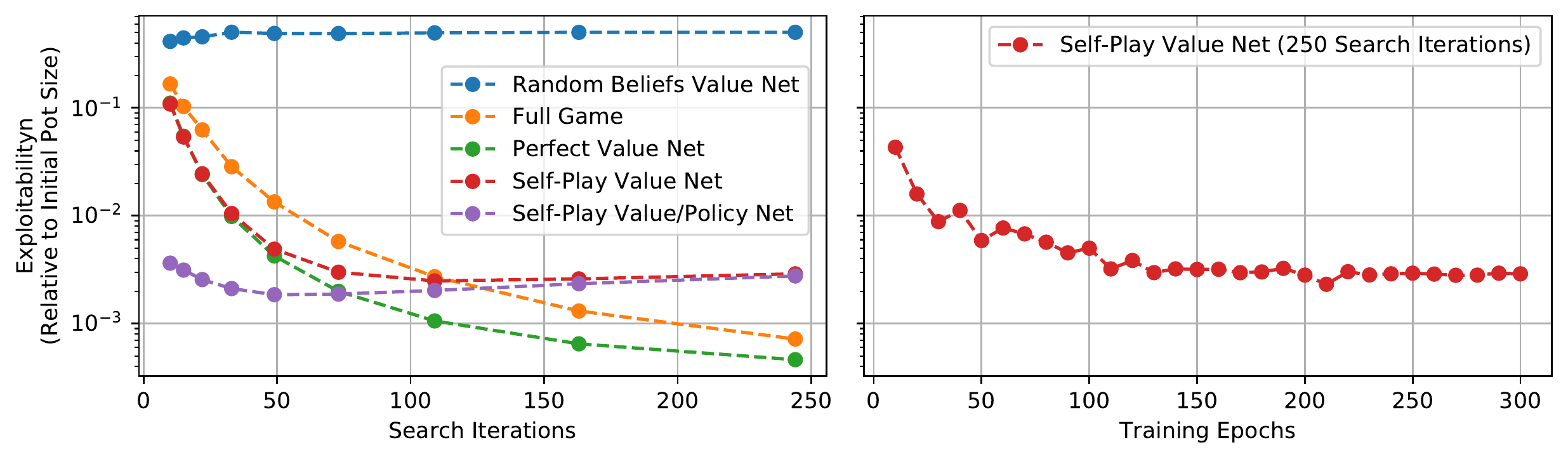}
    \vspace{-0.05in}
	\caption{\small{Convergence of different techniques in TEH. All subgames are solved using CFR-AVG. Perfect Value Net uses an oracle function to return the exact value of leaf nodes on each iteration. Self-Play Value Net uses a value function trained through self play. Self-Play Value/Policy Net additionally uses a policy network to warm start CFR. Random Beliefs trains the value net by sampling PBSs at random.
	}}
	\label{fig:turn_cfr}
\end{figure}

Table~\ref{tab:hunl} shows results for ReBeL in HUNL. We compare ReBeL to BabyTartanian8~\cite{brown2016baby} and Slumbot, prior champions of the Computer Poker Competition, and to the local best response (LBR)~\cite{lisy2017eqilibrium} algorithm. We also present results against Dong Kim, a top human HUNL expert that did best among the four top humans that played against Libratus. Kim played 7,500 hands. Variance was reduced by using AIVAT~\cite{burch2018aivat}. ReBeL played faster than 2 seconds per hand and never needed more than 5 seconds for a decision.

\begin{table*}[ht]
\begin{center}
\begin{tabular}{ l | c c c c }
\toprule
{\bf Bot Name} & {\bf Slumbot} & {\bf BabyTartanian8~\cite{brown2016baby}} & {\bf LBR~\cite{lisy2017eqilibrium}} & {\bf Top Humans}\\
\midrule
DeepStack~\cite{moravvcik2017deepstack}  & - & - & $383 \pm 112$ & -\\
\midrule
Libratus~\cite{brown2017superhuman}  & - & 63 $\pm$ 14 & - & 147 $\pm$ 39\\
\midrule
Modicum~\cite{brown2018depth}  & 11 $\pm$ 5 & 6 $\pm$ 3 & - & -\\
\midrule
\midrule
ReBeL \textit{(Ours)} & 45 $\pm$ 5 & 9 $\pm$ 4 & 881 $\pm$ 94 & 165 $\pm$ 69\\
\bottomrule
\end{tabular}
\end{center}
\vspace{-0.05in}
\caption{\small{Head-to-head results of our agent against benchmark bots BabyTartanian8 and Slumbot, as well as top human expert Dong Kim, measured in thousandths of a big blind per game. We also show performance against LBR~\cite{lisy2017eqilibrium} where the LBR agent must call for the first two betting rounds, and can either fold, call, bet 1$\times$ pot, or bet all-in on the last two rounds. The $\pm$ shows one standard deviation. For Libratus, we list the score against all top humans in aggregate; Libratus beat Dong Kim by 29 with an estimated $\pm$ of 78.
}}
\label{tab:hunl}
\vspace{-0.05in}
\end{table*}

Beyond just poker, Table~\ref{tab:liars} shows ReBeL also converges to an approximate Nash in several versions of Liar's Dice. Of course, tabular CFR does better than ReBeL when using the same number of CFR iterations, but tabular CFR quickly becomes intractable to run as the game grows in size.

\begin{table*}[h!]
\begin{center}
\begin{tabular}{l|rrrr}
\toprule
\textbf{Algorithm} &\textbf{ 1x4f} & \textbf{1x5f} & \textbf{1x6f} & \textbf{2x3f} \\
\midrule
Full-game FP  & 0.012 & 0.024 & 0.039 & 0.057 \\
Full-game CFR & 0.001 & 0.001 & 0.002 & 0.002 \\
\midrule
ReBeL FP & 0.041 & 0.020 & 0.040 & 0.020 \\
ReBeL CFR-D & 0.017 & 0.015 & 0.024 & 0.017 \\
\bottomrule
\end{tabular}
\end{center}
\vspace{-0.05in}
\caption{\small{Exploitability of different algorithms of 4 variants of Liar's Dice: 1 die with 4, 5, or 6 faces and 2 dice with 3 faces. The top two rows represent baseline numbers when a tabular version of the algorithms is run on the entire game for 1,024 iterations. The bottom 2 lines show the performance of ReBeL operating on subgames of depth 2 with 1,024 search iterations. For exploitability computation of the bottom two rows, we averaged the policies of 1,024 playthroughs and thus the numbers are upper bounds on exploitability.}}
\label{tab:liars}
\vspace{-0.05in}
\end{table*}

%% file: core/8_conclusions.tex
\vspace{-0.05in}
\section{Conclusions}
\vspace{-0.1in}
We present ReBeL, an algorithm that generalizes the paradigm of self-play reinforcement learning and search to imperfect-information games. We prove that ReBeL computes an approximate Nash equilibrium in two-player zero-sum games, demonstrate convergence in Liar's Dice, and demonstrate that it produces superhuman performance in the benchmark game of heads-up no-limit Texas hold'em.

ReBeL has some limitations that present avenues for future research. Most prominently, the input to its value and policy functions currently grows linearly with the number of infostates in a public state. This is intractable in games such as Recon Chess~\cite{newman2016reconnaissance} that have strategic depth but very little common knowledge. ReBeL's theoretical guarantees are also limited only to two-player zero-sum games.

Nevertheless, ReBeL achieves low exploitability in benchmark games and superhuman performance in heads-up no-limit Texas hold'em while leveraging far less expert knowledge than any prior bot. We view this as a major step toward developing universal techniques for multi-agent interactions.

%% file: core/9_impact.tex
\vspace{-0.05in}
\section*{Broader Impact}
\vspace{-0.1in}
We believe ReBeL is a major step toward general equilibrium-finding algorithms that can be deployed in large-scale multi-agent settings while requiring relatively little domain knowledge. There are numerous potential future applications of this work, including auctions, negotiations, cybersecurity, and autonomous vehicle navigation, all of which are imperfect-information multi-agent interactions.

The most immediate risk posed by this work is its potential for cheating in recreational games such as poker. While AI algorithms already exist that can achieve superhuman performance in poker, these algorithms generally assume that participants have a certain number of chips or use certain bet sizes. Retraining the algorithms to account for arbitrary chip stacks or unanticipated bet sizes requires more computation than is feasible in real time. However, ReBeL can compute a policy for arbitrary stack sizes and arbitrary bet sizes in seconds.

Partly for this reason, we have decided not to release the code for poker. We instead open source our implementation for Liar's Dice, a recreational game that is not played as competitively by humans. The implementation in Liar's Dice is also easier to understand and the size of Liar's Dice can be more easily adjusted, which we believe makes the game more suitable as a domain for research.

%% file: core/A0_contributions.tex
\section{List of contributions}

This paper makes several contributions, which we summarize here.

\begin{itemize}
    \item \textbf{RL+Search in two-player zero-sum imperfect-information games.} Prior work has developed RL+Search for two-player zero-sum perfect-information games. There has also been prior work on learning value functions in fully cooperative imperfect-information games~\cite{dibangoye2016optimally} and limited subsets of zero-sum imperfect-information games~\cite{horak2019solving}. However, we are not aware of any prior RL+Search algorithms for two-player zero-sum games in general. We view this as the central contribution of this paper.
    
    \item \textbf{Alternative to safe search techniques.} Theorem~\ref{th:safe} proves that, when doing search at test time with an accurate PBS value function, one can empirically play according to a Nash equilibrium by sampling a random iteration and passing down the beliefs produced by that iteration's policy. This result applies regardless of how the value function was trained and therefore applies to earlier techniques that use a PBS value function, such as DeepStack~\cite{moravvcik2017deepstack}.
    
    \item \textbf{Subgame decomposition via CFR-AVG.} We describe the CFR-AVG algorithm in Appendix~\ref{sec:cfrave}. CFR-D~\cite{burch2014solving} is a way to conduct depth-limited solving of a subgame with CFR when given a value function for PBSs. CFR-D is theoretically sound but has certain properties that may reduce performance in a self-play setting. CFR-AVG is a theoretically sound alternative to CFR-D that does not have these weaknesses. However, in order to implement CFR-AVG efficiently, in our experiments we modify the algorithm in a way that is not theoretically sound but empirically performs well in poker. Whether or not this modified form of CFR-AVG is theoretically sound remains an open question.
    
    \item \textbf{Connection between PBS gradients and infostate values.} Theorem~\ref{th:cfv} proves that all of the algorithms described in this paper can in theory be conducted using only $V_1$, not $\hat{v}$ (that is, a value function that outputs a single value for a PBS, rather than a vector of values for the infostates in the PBS). While this connection does not have immediate practical consequences, it does point toward a way of deploying the ideas in this paper to settings with billions or more infostates per PBS.
    
    \item \textbf{Fictitious Linear Optimistic Play (FLOP).} Section~\ref{sec:fp} introduces FLOP, a novel variant of Fictitious Play that is inspired by recent work on regret minimization algorithms~\cite{brown2019solving}. We show that FLOP empirically achieves near-$O(\frac{1}{T})$ convergence in the limit in both poker and Liar's Dice, does far better than any previous variant of FP, and in some domains is a reasonable alternative to CFR.
    
\end{itemize}

%% file: core/A1_code.tex
\section{Pseudocode for ReBeL}
\label{sec:code}

Algorithm~\ref{alg:rl_loop2} presents ReBeL in more detail.

We define the average of two policies to be the policy that is, in expectation, identical to picking one of the two policies and playing that policy for the entire game. Formally, if $\pi = \alpha \pi_1 + (1 - \alpha) \pi_2$, then $\pi(s_i) = \frac{\left( x_i^{\pi_1}(s_i) \alpha \right) \pi_1(s_i) + \left( x_i^{\pi_2}(s_i) (1 - \alpha) \right) \pi_2(s_i)}{x_i^{\pi_1}(s_i) \alpha + x_i^{\pi_2}(s_i) (1 - \alpha)}$ where $x_i^{\pi_1}(s_i)$ is the product of the probabilities for all agent~$i$ actions leading to $s_i$. Formally, $x_i^\pi(s_i)$ of infostate $s_i = (O_i^0,a_i^0,O_i^1,a_i^1,...,O_i^t)$ is $x_i^\pi(s_i) = \Pi_{t}(a_i^t)$.

\begin{algorithm}
\caption{ReBeL}

\begin{algorithmic}
\Function{Rebel-Linear-CFR-D}{$\beta_r, \theta^v, \theta^\pi, D^v,D^\pi$} \Comment{$\beta_r$ is the current PBS}
    \While{!\Call{IsTerminal}{$\beta_r$}}
    \State $G \gets $ \Call{ConstructSubgame}{$\beta_r$}
    \State $\bar{\pi}, \pi^{t_{\text{warm}}} \gets $ \Call{InitializePolicy}{$G, \theta^\pi$} \Comment{$t_{\text{warm}} = 0$ and $\pi^0$ is uniform if no warm start}
    \State $G \gets $ \Call{SetLeafValues}{$\beta_r, \pi^{t_{\text{warm}}}, \theta^v$}
    \State $v(\beta_r) \gets $ \Call{ComputeEV}{$G,\pi^{t_\text{warm}}$}
    \State $t_{\textit{sample}} \sim \mathrm{linear}\{t_{\text{warm}}+1, T\}$ \Comment{Probability of sampling iteration $t$ is proportional to $t$}
    \For{$t=(t_{\text{warm}}+1)..T$}
        \If{$t = t_{\textit{sample}}$}
            \State $\beta'_r \gets $ \Call{SampleLeaf}{$G,\pi^{t-1}$} \Comment{Sample one or multiple leaf PBSs}
        \EndIf
        \State $\pi^{t} \gets $ \Call{UpdatePolicy}{$G, \pi^{t-1}$}
        \State $\bar{\pi} \gets \frac{t}{t+2} \bar{\pi} + \frac{2}{t+2} \pi^{t}$
        \State $G \gets $ \Call{SetLeafValues}{$\beta_r, \pi^{t}, \theta^v$}
        \State $v(\beta_r) \gets \frac{t}{t+2} v(\beta_r) + \frac{2}{t+2}$ \Call{ComputeEV}{$G,\pi^{t}$}
    \EndFor
    \State Add $\{\beta_r,v(\beta_r)\}$ to $D^v$ \Comment{Add to value net training data}
    \For{$\beta \in G$} \Comment{Loop over the PBS at every public state in $G$}
    \State Add $\{\beta,\bar{\pi}(\beta)\}$ to $D^\pi$ \Comment{Add to policy net training data (optional)}
    \EndFor
    $\beta_r \gets \beta'_r$
    \EndWhile
\EndFunction
\State
\Function{SetLeafValues}{$\beta,\pi,\theta^v$}
\If{\Call{IsLeaf}{$\beta$}}
\For{$s_i \in \beta$} \Comment{For each infostate $s_i$ corresponding to $\beta$}
\State $v(s_i) = \hat{v}(s_i | \beta, \theta^v)$
\EndFor
\Else
\For{$a \in \mathcal{A}(\beta)$}
\State \Call{SetLeafValues}{$\mathcal{T}(\beta,\pi,a),\pi,\theta^v$}
\EndFor
\EndIf
\EndFunction
\State
\Function{SampleLeaf}{$G,\pi$}
    \State $i^* \sim \mathrm{unif}\{1,N\}$, $h \sim \beta_r$ \Comment{Sample a history randomly from the root PBS and a random player}
    \While{!\Call{IsLeaf}{$h$}}
        \State $c \sim \mathrm{unif}[0, 1]$
        \For{$i=1..N$}
            \If{$i == i^*$ and $c < \epsilon$} \Comment{we set $\epsilon = 0.25$ during training, $\epsilon = 0$ at test time}
                \State sample an action $a_i$ uniform random
            \Else{}
                \State sample an action $a_i$ according to $\pi_i(s_i(h))$
            \EndIf
        \EndFor
        \State $h \sim \tau(h,a)$
    \EndWhile
    \State \Return $\beta_{h}$ \Comment{Return the PBS corresponding to leaf node $h$}
\EndFunction
\end{algorithmic}
\label{alg:rl_loop2}
\end{algorithm}

%% file: core/A0_rules.tex
\section{Description of Games used for Evaluation}
\label{sec:games}

\subsection{Heads-up no-limit Texas hold'em poker (HUNL)}
HUNL is the two-player version of no-limit Texas hold'em poker, which is the most popular variant of poker in the world. For each ``hand'' (game) of poker, each player has some number of chips (the \emph{stack}) in front of them. In our experiments, stack size varies during training between \$5,000 and \$25,000 but during testing is always \$20,000, as is standard in the AI research community. Before play begins, Player~1 commits a \emph{small blind} of \$50 to the pot and Player~2 commits a \emph{big blind} of \$100 to the pot.

Once players commit their blinds, they receive two private cards from a standard 52-card deck. The first of four rounds of betting then occurs. On each round of betting, players take turns deciding whether to fold, call, or raise. If a player folds, the other player receives the money in the pot and the hand immediately ends. If a player calls, that player matches the opponent's number of chips in the pot. If a player raises, that player adds more chips to the pot than the opponent. The initial raise of the round must be at least \$100, and every subsequent raise on the round must be at least as large as the previous raise. A player cannot raise more than either player's stack size. A round ends when both players have acted in the round and the most recent player to act has called. Player~1 acts first on the first round. On every subsequent round, player~2 acts first.

Upon completion of the first round of betting, three \emph{community} cards are publicly revealed. Upon completion of the second betting round, another community card is revealed, and upon completion of the third betting round a final fifth community card is revealed. After the fourth betting round, if no player has folded, then the player with the best five-card poker hand, formed from the player's two private cards and the five community cards, is the winner and takes the money in the pot. In case of a tie, the money is split.

\subsection{Turn endgame hold'em (TEH)}
TEH is identical to HUNL except both players automatically call for the first two betting rounds, and there is an initial \$1,000 per player in the pot at the start of the third betting round. We randomize the stack sizes during training to be between \$5,000 and \$50,000 per player. The action space of TEH is reduced to at most three raise sizes ($0.5\times$ pot, $1\times$ pot, or all-in for the first raise in a round, and $0.75\times$ pot or all-in for subsequent raises), but the raise sizes for non-all-in raises are randomly perturbed by up to $\pm 0.1\times$ pot each game during training. Although we train on randomized stack sizes, bet sizes, and board cards, we measure exploitability on the case of both players having \$20,000, unperturbed bet sizes, and the first four board cards being $3 \spadesuit$7$\heartsuit$T$\diamondsuit$K$\spadesuit$. In this way we can train on a massive game while still measuring NashConv tractably. Even without the randomized stack and bet sizes, TEH has roughly $2 \cdot 10^{11}$ infostates.

\subsection{Liar's Dice}
Liar's Dice is a two-player zero-sum game in our experiments, though in general it can be played with more than two players. At the beginning of a game each player privately rolls $d$ dice with $f$ faces each. After that a betting stage starts where players take turns trying to predict how many dice of a specific kind there are among all the players, e.g., 4 dice with face 5. A player's bid must either be for more dice than the previous player's bid, or the same number of dice but a higher face. The round ends when a player challenges the previous bid (a call of \textit{liar}). If all players together have at least as many dice of the specified face as was predicted by the last bid, then the player who made the bid wins. Otherwise the player who challenged the bid wins. We use the highest face as a \textit{wild} face, i.e., dice with this face count towards a bid for any face.

%% file: core/A5_domain_knowledge.tex
\section{Domain Knowledge Leveraged in our Poker AI Agent}
\label{sec:dk}

The most prominent form of domain knowledge in our ReBeL poker agent is the simplification of the action space during self play so that there are at most 8 actions at each decision point. The bet sizes are hand-chosen based on conventional poker wisdom and are fixed fractions of the pot, though each bet size is perturbed by $\pm 0.1\times$ pot during training to ensure diversity in the training data.

We specifically chose not to leverage domain knowledge that has been widely used in previous poker AI agents:
\begin{itemize}
    \item All prior top poker agents, including DeepStack~\cite{moravvcik2017deepstack}, Libratus~\cite{brown2017superhuman}, and Pluribus~\cite{brown2019superhuman}, have used \emph{information abstraction} to bucket similar infostates together based on domain-specific features~\cite{johanson2012finding,ganzfried2014potential,brown2015hierarchical}. Even when computing an exact policy, such as during search or when solving a poker game in its entirety~\cite{gilpin2005optimal,bowling2015heads}, past agents have used \emph{lossless abstraction} in which strategically identical infostates are bucketed together. For example, a flush of spades may be strategically identical to a flush of hearts.
    
    Our agent does not use any information abstraction, whether lossy or lossless. The agent computes a unique policy for each infostate. The agent's input to its value and policy network is a probability distribution over pairs of cards for each player, as well as all public board cards, the amount of money in the pot relative to the stacks of the players, and a flag for whether a bet has occurred on this betting round yet.
    
    \item DeepStack trained its value network on random PBSs. In addition to reducing the dimensionality of its value network input by using information abstraction, DeepStack also sampled PBSs according to a handcrafted algorithm that would sample more realistic PBSs compared to sampling uniform random. We show in Section~\ref{sec:results} that training on PBSs sampled uniformly randomly without information abstraction results in extremely poor performance in a value network.
    
    Our agent collects training data purely from self play without any additional heuristics guiding which PBSs are sampled, other than an exploration hyperparameter that was set to $\epsilon = 0.25$ in all experiments.
    
    \item In cases where both players bet all their chips before all board cards are revealed, past poker AIs compute the exact expected value of all possible remaining board card outcomes. This is expensive to do in real time on earlier rounds, so past agents pre-compute this expected value and look it up during training and testing. Using the exact expected value reduces variance and makes learning easier.
    
    Our agent does not use this shortcut. Instead, the agent learns these ``all-in'' expected values on its own. When both agents have bet all their chips, the game proceeds as normal except neither player is allowed to bet.
    
    \item The search space in DeepStack~\cite{moravvcik2017deepstack} extends to the start of the next betting round, except for the third betting round (out of four) where it instead extends to the end of the game. Searching to the end of the game on the third betting round was made tractable by using information abstraction on the fourth betting round (see above). Similarly, Libratus~\cite{brown2017safe}, Modicum~\cite{brown2018depth}, and Pluribus~\cite{brown2019superhuman} all search to the end of the game when on the third betting round. Searching to the end of the game has the major benefit of not requiring the value network to learn values for the end of the third betting round. Thus, instead of the game being three ``levels'' deep, it is only two levels deep. This reduces the potential for propogation of errors.
    
    Our agent always solves to the end of the current betting round, regardless of which round it is on.
    
    \item The depth-limited subgames in DeepStack extended to the start of the next betting round on the second betting round. On the first betting round, it extended to the end of the first betting round for most of training and to the start of the next betting round for the last several CFR iterations. Searching to the start of the next betting round was only tractable due to the abstractions mentioned previously and due to careful optimizations, such as implementing CFR on a GPU.
    
    Our agent always solves to the end of the current betting round regardless of which round it is on. We implement CFR only on a single-thread CPU and avoid any abstractions. Since a subgame starts at the beginning of a betting round and ends at the start of the next betting round, our agent must learn six ``layers'' of values (end of first round, start of second round, end of second round, start of third round, end of third round, start of fourth round) compared to three for DeepStack (end of first round, start of second round, start of third round).
    
    \item DeepStack used a separate value network for each of the three ``layers'' of values (end of first round, start of second round, start of third round). Our agent uses a single value network for all situations.
\end{itemize}

%% file: core/A7_reprodicubility.tex
\section{Hyper parameters}\label{sec:hyperparams}
In this section we provide details of the value and policy networks and the training procedures. 

We approximate the value and policy functions using artificial neural networks. The input to the value network consists of three components for both games: agent index, representation of the public state, and a probability distribution over infostates for both agents. For poker, the public state representation consists of the board cards and the common pot size divided by stack size; for Liar's Dice it is the last bid and the acting agent. The output of the network is a vector of values for each possible infostate of the indexed agent, e.g., each possible poker hand she can hold.

We trained a policy network only for poker.
The policy network state representation additionally contains pot size fractions for both agents separately as well as a flag for whether there have been any bets so far in the round. The output is a probability distribution over the legal actions for each infostate.

As explained in section~\ref{sec:training} we use Multilayer perceptron with GeLU~\cite{hendrycks2016gaussian} activation functions and LayerNorm~\cite{ba2016layer} for both value and policy networks. 

For poker we represent the public state as a concatenation of a vector of indices of the board cards, current pot size relative to the stack sizes, and binary flag for the acting player. The size of the full input is
\begin{equation}
1(\text{agent~index})+ 1 (\text{acting~agent}) + 1 (\text{pot}) + 5 (\text{board}) + 2 \times 1326 (\text{infostate~beliefs})
\nonumber
\end{equation}
We use card embedding for the board cards similar to~\cite{brown2019deep} and then apply MLP. Both the value and the policy networks contain 6 hidden layers with 1536 layers each. For all experiments we set the probability to explore a random action to $\epsilon=25\%$ (see Section~\ref{sec:method_high}). To store the training data we use a simple circular buffer of size 12M and sample uniformly. Since our action abstraction contains at most 9 legal actions, the size of the target vector for the policy network is 9 times bigger than one used for the value network. In order to make it manageable, we apply linear quantization to the policy values. As initial data is produced with a random value network, we remove half of the data from the replay buffer after 20 epochs.

For the full game we train the network with Adam optimizer with learning rate $3\times 10^{-4}$ and halved the learning rate every 800 epochs. One epoch is 2,560,000 examples and the batch size 1024. We used 90 DGX-1 machines, each with 8 $\times$ 32GB Nvidia V100 GPUs for data generation. We report results after 1,750 epochs.
For TEH experiments we use higher initial learning rate $4 \times 10^{-4}$, but halve it every 100 epochs.  We report results after 300 epochs.

For Liar's Dice we represent the state as a concatenation of a one hot vector for the last bid and binary flag for the acting player. The size of the full input is $$1(\text{agent~index}) + 1 (\text{acting~agent}) + n_\text{dice} n_\text{faces} (\text{last~bid}) + 2 {n_\text{faces}}^{n_\text{dice}} (\text{infostate~beliefs}).$$ The value network contains 2 hidden layers with 256 layers each. We train the network with Adam optimizer with learning rate $3\times 10^{-4}$ and halved the learning rate every 400 epochs. One epoch is 25,600 examples and the batch size 512. During both training and evaluation we run the search algorithm for 1024 iterations. We use single GPU for training and 60 CPU threads for data generation. We trained the network for 1000 epochs. To reduce the variance in RL+Search results, we evaluated the three last checkpoints and reported averages in table~\ref{tab:liars}.

\subsection{Human Experiments for HUNL}

We evaluated our HUNL agent against Dong Kim, a top human professional specializing in HUNL. Kim was one of four humans that played against Libratus~\cite{brown2017superhuman} in the man-machine competition which Libratus won. Kim lost the least to Libratus. However, due to high variance, it is impossible to statistically compare the performance of the individual humans that participated in the competition.

A total of 7,500 hands were played between Kim and the bot. Kim was able to play from home at his own pace on any schedule he wanted. He was also able to play up to four games simultaneously against the bot. To incentivize strong play, Kim was offered a base compensation of $\$1 \pm \$0.05x $ for each hand played, where $x$ signifies his average win/loss rate in terms of big blinds per hundred hands played. Kim was guaranteed a minimum of \$0.75 per hand and could earn no more than \$2 per hand. Since final compensation was based on the variance-reduced score rather than the raw score, Kim was not aware of his precise performance during the experiment.

The bot played at an extremely fast pace. No decision required more than 5 seconds, and the bot on average plays faster than 2 seconds per hand in self play. To speed up play even further, the bot cached subgames it encountered on the preflop. When the same subgame was encountered again, it would simply reuse the solution it had already computed previously.

Kim's variance-reduced score, which we report in Section~\ref{sec:results}, was a loss of $165 \pm 69$ where the $\pm$ indicates one standard error. His raw score was a loss of $358 \pm 188$.

%% file: core/A8.1_dvalue_proof.tex
\section{Proof Related to Value Functions (Theorem~\ref{th:cfv})}
\label{app:dvalue_proof}

We start by proving some preliminary Lemmas. For simplicity, we will sometimes prove results for only one player, but the results hold WLOG for both players.

For some policy profile $\pi=(\pi_1, \pi_2)$, let $v_i^\pi(\beta): \mathcal{B} \to \mathbb{R}^{|S_i|}$ be a function that takes as input a PBS and outputs infostate values for player~$i$.

\begin{lemma}
\label{lemma:v_linear}
Let $V_1^{\pi_2}(\beta)$ be player 1's value at $\beta$ assuming that player~2 plays $\pi_2$ in a 2p0s game. $V_1^{\pi_2}(\beta)$ is linear in $\beta_1$.
\end{lemma}

\begin{proof}
This follows directly from
the definition of $v_i^{\pi}(s_i|\beta)$
along with the definition of $V_1$, $$V_1^{\pi_2}(\beta) = \sum_{s_1 \in S_1(s_\text{pub})} {\beta_1(s_1) v_1(s_1|\beta, (BR(\pi_2), \pi_2))}$$
\end{proof}

\begin{lemma}
\label{lemma:v_concave}
$V_1(\beta) = \min_{\pi_2}{V_1^{\pi_2}(\beta)}$, and the set of $\pi_2$ that attain $V_1(\beta)$ at $\beta_0$ are precisely the Nash equilibrium policies at $\beta_0$. This also implies that $V_1(\beta)$ is concave.
\end{lemma}

\begin{proof}
By definition, the Nash equilibrium at $\beta$ is the minimum among all choices of $\pi_2$ of the value to player 1 of her best response to $\pi_2$. Any $\pi_2$ that achieves this Nash equilibrium value when playing a best response is a Nash equilibrium policy.

From Lemma \ref{lemma:v_linear}, we know that each $V_1^{\pi_2}(\beta)$ is linear, which implies that $V_1(\beta)$ is concave since any function that is the minimum of linear functions is concave.
\end{proof}

\begin{figure}[!h]
    \label{fig:v1_diagram}
    \centering
	\includegraphics[width=100mm]{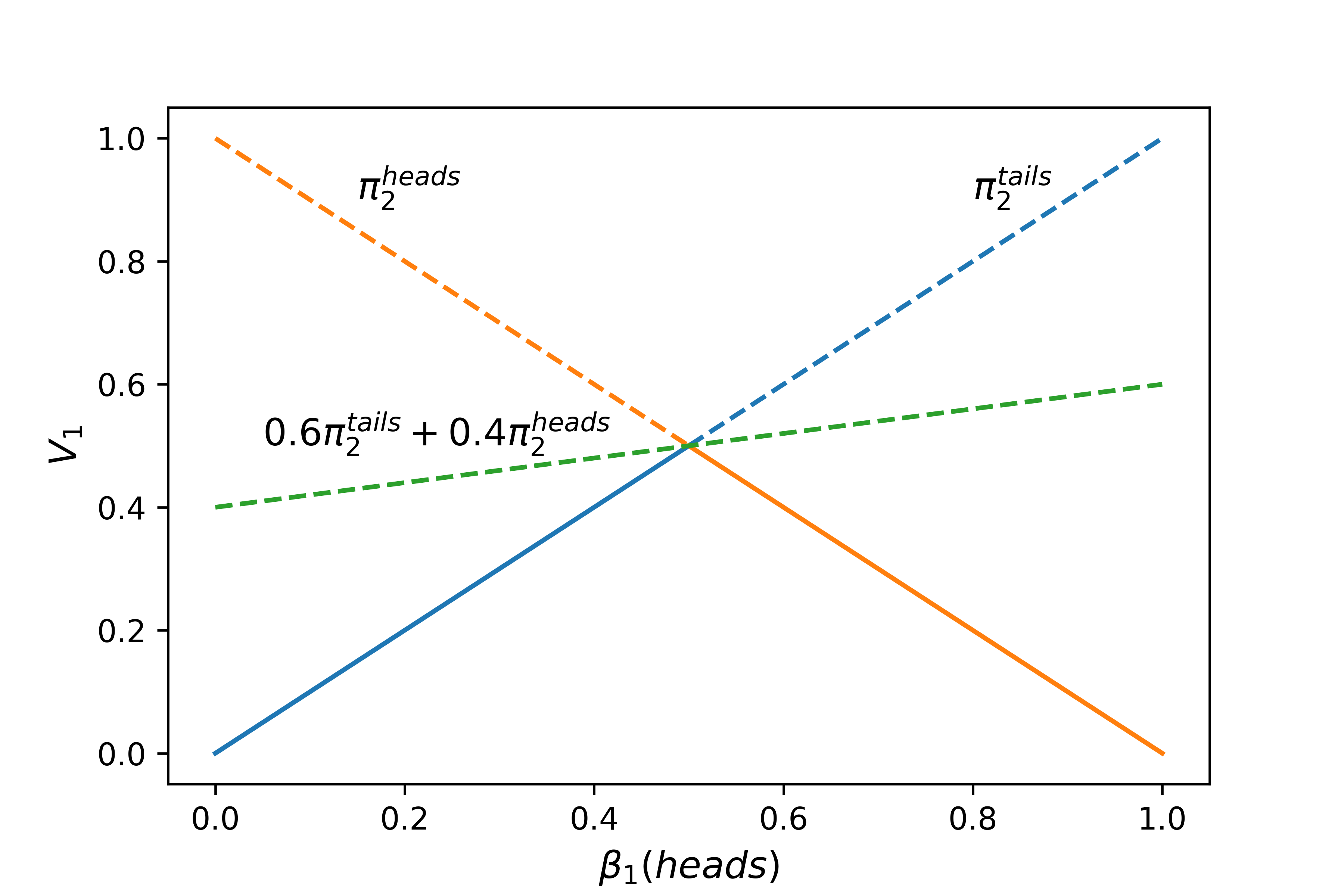}
	\caption{\small Illustration of Lemma \ref{lemma:v_concave}. In this simple example, the subgame begins with some probability $\beta(heads)$ of a coin being heads-up, which player 1 observes. Player 2 then guesses if the coin is heads or tails, and wins if he guesses correctly. The payoffs for Player 2's pure strategies are shown as the lines marked $\pi_2^{heads}$ and $\pi_2^{tails}$. The payoffs for a mixed strategy is a linear combination of the pure strategies. The value for player 1 is the minimum among all the lines corresponding to player 2 strategies, denoted by the solid lines.}
\end{figure}

Now we can turn to proving the Theorem. 

Consider a function $\tilde{V}_1$ that is an extension of $V_1$ to unnormalized probability distributions over $S_1$ and $S_2$; i.e. $\tilde{V}_i((s_{pub}, b_1, b_2)) = V_i((s_\text{pub}, b_1/|b_1|_1, b_2/|b_2|_1))$. $\tilde{V}_i = V_i$ on the simplex of valid beliefs, but we extend it in this way to $\mathbb{R}_{\geq 0}^{|s_1|} \setminus \vec{0}$ so that we can consider gradients w.r.t. $p(s_1)$.

We will use the term `supergradient' to be the equivalent of the subgradient for concave functions. Formally, $g$ is a supergradient of concave function $F$ at $x_0$ iff for any $x$ in the domain of $F$, $$F(x) - F(x_0) \leq g\cdot(x - x_0).$$ Also, $superg(F)=-subg(-F)$.

\begin{theorem-non}[Restatement of Theorem \ref{th:cfv}]
For any belief $\beta=(\beta_1,\beta_2)$ (for the beliefs over player 1 and 2 infostates respectively) and any policy $\pi^*$ that is a Nash equilibrium of the subgame rooted at $\beta$,

\begin{equation}
    v_1^{\pi^*}(s_1|\beta) = V_1(\beta) +\bar{g} \cdot \hat{s}_1
\end{equation} 
for some supergradient $\bar{g}$ of $\tilde{V}_1(\beta)$ with respect to $\beta_{1}$, where $\hat{s}_1$ is the unit vector in direction $s_1$.
\end{theorem-non}

\begin{proof}

Lemma \ref{lemma:v_concave} shows that $V_1(\beta)$ is a concave function in $\beta_1$, and its extension $\tilde{V}$ off the simplex is constant perpendicular to the simplex, so $\tilde{V}$ is concave as well. Therefore the notion of a supergradient is well-defined.

Now, consider some policy profile $\pi^*=(\pi_1^*,\pi_2^*)$ that is a Nash equilibrium of $G(\beta_0)$. $V_1^{\pi_2^*}$ is a linear function and tangent to $V_1$ at $\beta_0$; therefore, its gradient is a supergradient of $V_1$ at $\beta_0$. Its gradient is

\begin{align}
    \nabla_{\beta_1} V_1^{\pi_2^*}(\beta) & = \nabla_{\beta_1}  \sum_{s_1 \in S_1(s_\text{pub})} {\beta_1(s_1) v_1^{\pi_2^*}(s_1|\beta)} \label{eq:gradv} \\
    & = \sum_{s_1 \in S_1(s_\text{pub})} \hat{s}_1 v_1^{\pi_2^*}(s_1|\beta) + \beta(s_1) \nabla_{\beta_1} v^{\pi_2^*}(s_1|\beta) \\
    & = \sum_{s_1 \in S_1(s_\text{pub})} \hat{s}_1 v_1^{\pi_2^*}(s_1|\beta)
    \label{eq:grad_indep}
\end{align}

Equation \ref{eq:grad_indep} follows from the fact that for a fixed opponent policy $\pi_2^*$, each player 1 infostate $s_1$ is an independent MDP whose value $v_1^{\pi_2^*}(s_1|\beta)$ doesn't depend on the probabilities $\beta_1$ of being in the different infostates (although it might depend on $\beta_2$ since player 1 doesn't observe these).

Note also that the gradient of $V_1^{\pi_2^*}$ is correct even when some infostates $s_1$ have probability 0, due to the fact that we defined $v_1^{\pi_2}(s_1|\beta)$ as the value of player 1 playing a \textit{best response} to $\pi_2$ at each infostate $s_1$ (rather than just playing the equilibrium policy $\pi_1^*$, which may play arbitrarily at unvisited infostates).

Finally, let's compute $g \cdot \hat{s}_1$ at some $\beta_1$ on the simplex (i.e. $|\beta_1|_1=1$).

\begin{align}
    g &= 
    \nabla_{\beta_1/|\beta_1|_1} V_1^{\pi^*}(s_\text{pub}, \beta_1/|\beta_1|, \beta_2) 
    \cdot
    \frac{d}{d\beta_1} \left(\frac{\beta_1}{|\beta_1|_1} \right) \hspace{1cm} \text{(chain rule)}\\
    &=\left( \sum_{s_1' \in S_1(s_\text{pub})} \hat{s}_1' v_1^{\pi^*}(s_1'|\beta) \right) \cdot (|\beta_1|_1 - \beta_1)/(|\beta_1|_1)^2 \hspace{1cm} \text{(Eq. \ref{eq:gradv})}\\
    &=\left( \sum_{s_1' \in S_1(s_\text{pub})} \hat{s}_1' v_1^{\pi^*}(s_1'|\beta) \right) \cdot (1 - \beta_1)\hspace{1cm} \text{(since $|\beta_1|_1=1$)}\\
    &= \sum_{s_1' \in S_1(s_\text{pub})} \hat{s}_1' v_1^{\pi^*}(s_1'|\beta) - \sum_{s_1' \in S_1(s_\text{pub})} \beta_1(s_1')  v_1^{\pi^*}(s_1'|\beta) \\
    &= \sum_{s_1' \in S_1(s_\text{pub})} \hat{s}_1' v_1^{\pi^*}(s_1'|\beta) - V_1(\beta) \\
    g \cdot \hat{s}_1 & = v_1^{\pi^*}(s_1|\beta) - V_1(\beta)
\end{align}

And we're done.

\end{proof}

%% file: core/A8.2_subgame_solving_proofs.tex
\section{Proofs Related to Subgame Solving (Theorems \ref{th:perfect} and \ref{th:safe})}
\label{sec:proof_safe}

\begin{lemma}
\label{lemma:recursive_ev}
Running Algorithm~\ref{alg:rl_loop} for $N \rightarrow \infty$ times in a depth-limited subgame rooted at PBS~$\beta_r$ will compute an infostate value vector $v_i^{\pi^*}(\beta_r)$ corresponding to the values of the infostates when $\pi^*$ is played in the (not depth-limited) subgame rooted at $\beta_r$, where $\pi^*$ is a $\frac{C}{\sqrt{T}}$-Nash equilibrium for some constant $C$.
\end{lemma}

\begin{proof}
A key part of our proof is the insight that Algorithm~\ref{alg:rl_loop} resembles the CFR-D algorithm from~\cite{burch2014solving} if Algorithm~\ref{alg:rl_loop} were modified such that there was no random sampling and every call to the value network was replaced with a recursive call to the CFR-D algorithm.

Suppose the subgame rooted at $\beta_r$ extends to the end of the game and therefore there are no calls to the neural network. Then CFR is proven to compute a $\frac{C}{\sqrt{T}}$-Nash equilibrium~\cite{zinkevich2008regret}, which we will call $\pi^*$, and Algorithm~\ref{alg:rl_loop} will indeed learn a value vector $v_i^{\pi^*}(\beta_r)$ for $\beta_r$ corresponding to the infostate values of $\beta_r$ when $\pi^*$ is played in the subgame. Thus, the base case for the inductive proof holds.

Now suppose we have a depth-limited subgame rooted at $\beta_r$. Assume that for every leaf PBS $\beta_z^{\pi^t}$ for iterations $t \le T$, where $\pi^t$ is the policy in the subgame on iteration~$t$, that we have already computed an infostate value vector $v_i^{\pi^*}(\beta^{\pi^t}_z)$ corresponding to the values of the infostates in $\beta^{\pi^t}_z$ when $\pi^*$ is played in the subgame rooted at $\beta^{\pi^t}_z$, where $\pi^*$ is a $\frac{C}{\sqrt{T}}$-Nash equilibrium for some constant $C$. Assume also that Lemma~\ref{lemma:recursive_ev} holds for all leaf PBSs $\beta_z^{\pi^{T+1}}$.

First, for leaf PBSs that neither player reaches with positive probability, the values for the leaf PBS have no affect on CFR or the values computed for the root infostates because CFR weights the values by the probability the player reaches the PBS~\cite{zinkevich2008regret}.

Now consider a leaf PBS $\beta_z^{\pi^{T+1}}$ that some player reaches with positive probability. Since $v_i^{\pi^*}(\beta_z^{\pi^t})$ has already been computed for all $\beta_z^{\pi^t}$ and since we are running a deterministic algorithm, so $v_i^{\pi^*}(\beta_z^{\pi^t})$ will not change with subsequent calls of Algorithm~\ref{alg:rl_loop} on $\beta_z^{\pi^t}$ for all $t \le T$. Thus, $\pi^t$ will be the same for all $t \le T$. Since Algorithm~\ref{alg:rl_loop} samples a random CFR iteration, and since the leaf PBS $\beta_z^{\pi^{T+1}}$ is sampled with positive probability for some player when iteration~$T+1$ is sampled, so the algorithm will sample $\beta_z^{\pi^{T+1}}$ $N'$ times, where $N' \rightarrow \infty$ as $N \rightarrow \infty$. Since Lemma~\ref{lemma:recursive_ev} holds for $\beta_z^{\pi^{T+1}}$, so eventually $v_i^{\pi^*}(\beta_z^{\pi^{T+1}})$ will be computed for $\beta_z^{\pi^{T+1}}$. Therefore, due to CFR-D~\cite{burch2014solving}, Lemma~\ref{lemma:recursive_ev} will hold for $\beta_r$ and the inductive step is proven.
\end{proof}

\begin{theorem-non}[Restatement of Theorem \ref{th:perfect}]
Consider an idealized value approximator that returns the most recent sample of the value for sampled PBSs, and 0 otherwise. Running Algorithm~\ref{alg:rl_loop} with $T$ iterations of CFR in each subgame will produce a value approximator that produces values that correspond to a $\frac{C}{\sqrt{T}}$-equilibrium policy for any PBS that could be encountered during play, where $C$ is a game-dependent constant.
\end{theorem-non}

\begin{proof}
Since we run Algorithm~\ref{alg:rl_loop} for $N \rightarrow \infty$ times at the root of the game, so by Lemma~\ref{lemma:recursive_ev}, Theorem~\ref{th:perfect} is true.
\end{proof}

\begin{theorem-non}[Restatement of Theorem~\ref{th:safe}]
If Algorithm~\ref{alg:rl_loop} is run at test time with no off-policy exploration, a value network that has error at most $\delta$ for any leaf PBS, and with $T$ iterations of CFR being used to solve subgames, then the algorithm plays a $(\delta C_1 + \frac{\delta C_2}{\sqrt{T}})$-Nash equilibrium, where $C_1, C_2$ are game-specific constants. 
\end{theorem-non}

\begin{proof}
We prove the theorem inductively. Consider first a subgame near the end of the game that is not depth-limited. I.e., it has no leaf nodes. Clearly, the policy $\pi^*$ that Algorithm~\ref{alg:rl_loop} using CFR plays in expectation is a $\frac{k_1}{\sqrt{T}}$-Nash equilibrium for game-specific constant $k_1$ in this subgame.

Rather than play the average policy over all $T$ iterations~$\bar{\pi}^T$, one can equivalently pick a random iteration $t \sim \mathrm{uniform}\{1,T\}$ and play according to~$\pi^t$, the policy on iteration~$t$. This algorithm is also a $\frac{k_1}{\sqrt{T}}$-Nash equilibrium in expectation.

Next, consider a depth-limited subgame $G$ such that for any leaf PBS $\beta^t$ on any CFR iteration~$t$, the policy that Algorithm~\ref{alg:rl_loop} plays in the subgame rooted at $\beta^t$ is in expectation a $\delta$-Nash equilibrium in the subgame. If one computes a policy for $G$ using tabular CFR-D~\cite{burch2014solving} (or, as discussed in Section~\ref{sec:cfrave}, using CFR-AVG), then by Theorem~2 in~\cite{burch2014solving}, the average policy over all iterations is $k_2 \delta + \frac{k_3}{\sqrt{T}}$-Nash equilibrium.

Just as before, rather than play according to this average policy~$\bar{\pi}^T$, one can equivalently pick a random iteration $t \sim \mathrm{uniform}\{1,T\}$ and play according to $\pi^t$. Doing so would also result in a $k_2 \delta + \frac{k_3}{\sqrt{T}}$-Nash equilibrium in expectation. This is exactly what Algorithm~\ref{alg:rl_loop} does.

Since there are a finite number of ``levels'' in a game, which is a game-specific constant, Algorithm~\ref{alg:rl_loop} plays according to a $\delta C_1 + \frac{\delta C_2}{\sqrt{T}}$-Nash equilibrium.

\end{proof}

%% file: core/A1_fp.tex
\section{Fictitious Linear Optimistic Play}
\label{sec:fp}
Fictitious Play (FP)~\cite{brown1951iterative} is an extremely simple iterative algorithm that is proven to converge to a Nash equilibrium in two-player zero-sum games. However, in practice it does so at an extremely slow rate. On the first iteration, all agents choose a uniform policy $\pi_i^0$ and the average policy $\bar{\pi}^0_i$ is set identically. On each subsequent iteration~$t$, agents compute a best response to the other agents' average policy $\pi_i^t = \argmax_{\pi_i}v_i(\pi_i, \bar{\pi}_{-i}^{t-1})$ and update their average policies to be $\bar{\pi}_i^t = \frac{t-1}{t} \bar{\pi}_i^{t-1} + \frac{1}{t} \pi_{i}^{t}$. As $t \to \infty$, $\bar{\pi}^t$ converges to a Nash equilibrium in two-player zero-sum games.

It has also been proven that a family of algorithms similar to FP known as \textbf{generalized weakened fictitious play (GWFP)} also converge to a Nash equilibrium so long as they satisfy certain properties~\cite{van2000weakened,leslie2006generalised}, mostly notably that in the limit the policies on each iteration converge to best responses.

In this section we introduce a novel variant of FP we call \textbf{Fictitious Linear Optimistic Play (FLOP)} which is a form of GWFP. FLOP is inspired by related variants in CFR, in particular Linear CFR~\cite{brown2019solving}.
FLOP converges to a Nash equilibrium much faster than FP while still being an extremely simple algorithm. However, variants of CFR such as Linear CFR and Discounted CFR~\cite{brown2019solving} still converge much faster in most large-scale games.

In FLOP, the initial policy $\pi_i^0$ is uniform. On each subsequent iteration~$t$, agents compute a best response to an \emph{optimistic}~\cite{chiang2012online,rakhlin2013online,syrgkanis2015fast} form of the opponent's average policy in which $\pi_{-i}^{t-1}$ is given extra weight: $\pi_i^t = \argmax_{\pi_i}v_i(\pi_i, \frac{t}{t+2}\bar{\pi}_{-i}^{t-1} + \frac{2}{t+2}\pi_{-i}^{t-1})$. The average policy is updated to be $\bar{\pi}_i^t = \frac{t - 1}{t + 1} \bar{\pi}_i^{t-1} + \frac{2}{t+1} \pi_i^t$. Theorem~\ref{th:flop} proves that FLOP is a form of GWFP and therefore converges to a Nash equilibrium as $t \to \infty$.

\begin{theorem}
\label{th:flop}
FLOP is a form of Generalized Weakened Fictitious Play.
\end{theorem}

\begin{proof}
Assume that the range of payoffs in the game is $M$. Since $\pi_i^t = \argmax_{\pi_i}v_i(\pi_i, \frac{t}{t+2}\bar{\pi}_{-i}^{t-1} + \frac{2}{t+2}\pi_{-i}^{t-1})$, so $\pi_i^t$ is an $\epsilon_t$-best response to $\bar{\pi}^{t-1}_{-i}$ where $\epsilon_t < M\frac{2}{t+2}$ and $\epsilon_t \to 0$ as $t \to \infty$. Thus, FLOP is a form of GWFP with $\alpha_t = \frac{2}{t}$.
\end{proof}

\begin{figure}[!h]
	\vspace{-0.1in}
	\centering
	\includegraphics[width=0.95\textwidth]{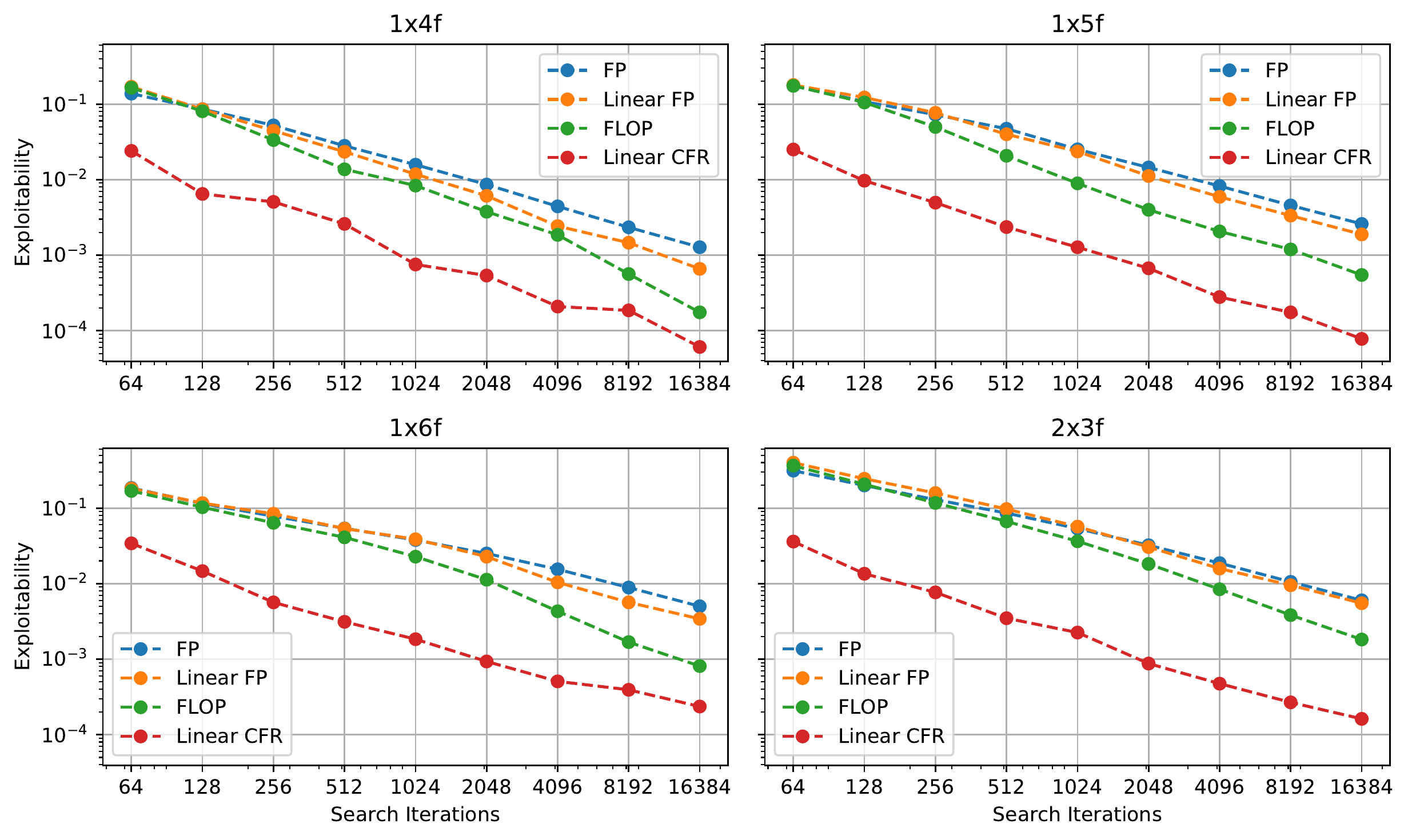}
	\vspace{-0.1in}

	\caption{\small{Exploitability of different algorithms of 4 variants of Liar's Dice: 1 die with 4, 5, or 6 faces and 2 dice with 3 faces. For all games FLOP outperforms Linear FP, but does not match the convergence of Linear CFR.
	}}
	\label{fig:flop_liars}
\end{figure}

\begin{figure}[!h]
	\vspace{-0.1in}
	\centering
	\includegraphics[width=0.95\textwidth]{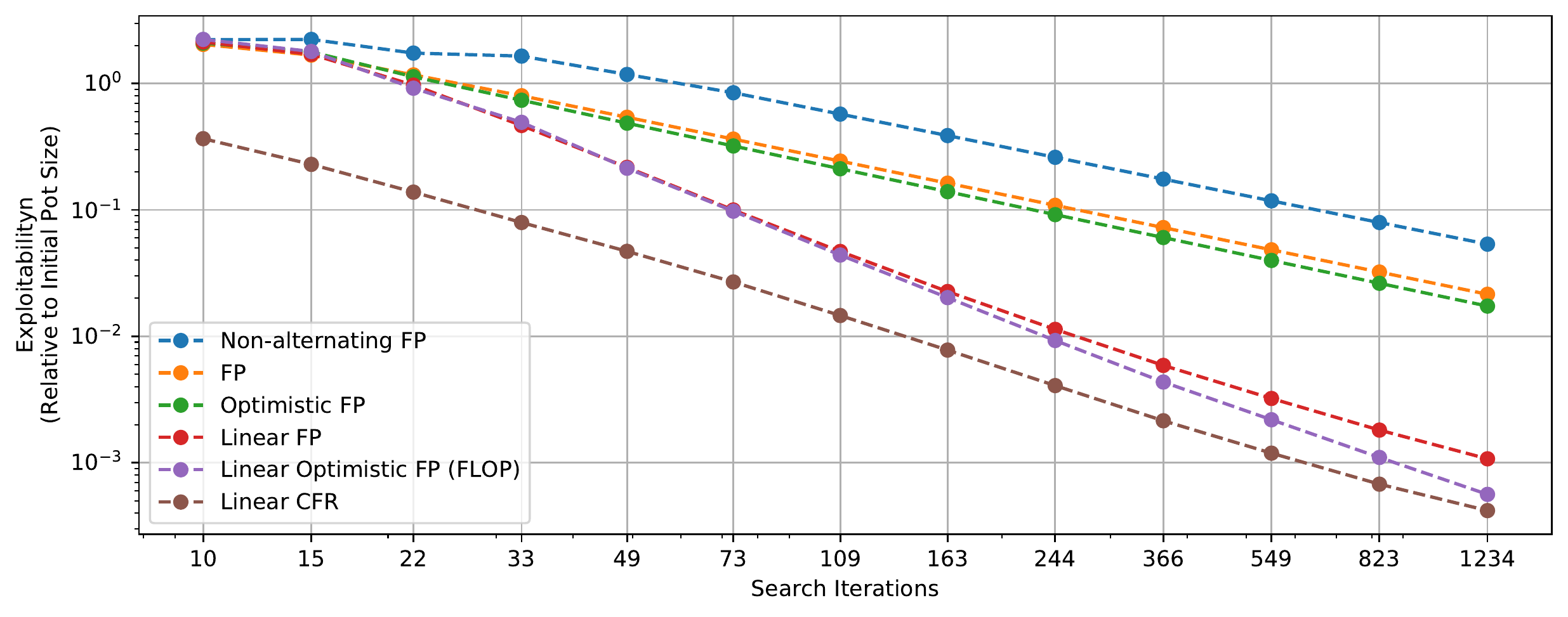}
	\vspace{-0.1in}

	\caption{\small{Exploitability of different algorithms for Turn Endgame Hold’em.}}
	\label{fig:flop_holdem}
\end{figure}

%% file: core/A2_decomp_proof.tex
\section{CFR-AVG: CFR Decomposition using Average Strategy}
\label{sec:cfrave}
On each iteration~$t$ of CFR-D, the value of every leaf node $z$ is set to $\hat{v}(s_i(z)|\beta^{\pi^t}_{z})$. Other than changing the values of leaf nodes every iteration, CFR-D is otherwise identical to CFR. If $T$ iterations of CFR-D are conducted with a value network that has error at most $\delta$ for each infostate value, then $\bar{\pi}^T$ has exploitability of at most $k_1 \delta + k_2 / \sqrt{T}$ where $k_1$ and $k_2$ are game-specific constants~\cite{moravvcik2017deepstack}.

Since it is the \emph{average} policy profile $\bar{\pi}^t$, not $\pi^t$, that converges to a Nash equilibrium as $t \to \infty$, and since the leaf PBSs are set based on $\pi^t$, the input to the value network~$\hat{v}$ may span the entire domain of inputs even as $t \rightarrow \infty$.
For example, suppose in a Nash equilibrium~$\pi^*$ the probability distribution at $\beta^{\pi^*}_{z}$ was uniform. Then the probability distribution at $\beta^{\pi^t}_{z}$ for any individual iteration~$t$ could be \emph{anything}, because regardless of what the probability distribution is, the average over all iterations could still be uniform in the end. Thus, $\hat{v}$ may need to be accurate over the entire domain of inputs rather than just the subspace near $\beta^{\pi^*}_{z}$.

In \textbf{CFR-AVG}, leaf values are instead set according to the \emph{average policy}~$\bar{\pi}^t$ on iteration~$t$. When a leaf PBS is sampled, the leaf node is sampled with probability determined by $\pi^t$, but the PBS itself is defined using $\bar{\pi}^t$.

We first describe the tabular form of CFR-D~\cite{burch2014solving}. Consider a game $G'$ and a depth-limited subgame $G$, where both $G'$ and $G$ share a root but $G$ extends only a limited number of actions into the future. Suppose that $T$ iterations of a modified form of CFR are conducted in $G'$. On each iteration $t \le T$, the policy $\pi(s_i)$ is set according to CFR for each $s_i \in G$. However, for every infostate $s'_i \in G' \setminus G$, the policy is set differently than what CFR would call for. At each leaf public state $s'_\text{pub}$ of $G$, we solve a subgame rooted at $\beta_{s'_{\text{pub}}}^{\pi^{t}}$ by running $T'$ iterations of CFR. For each $s'_i$ in the subgame rooted at $\beta_{s'_{\text{pub}}}^{\pi^{t}}$, we set $\pi^{t}(s'_i) = \bar{\pi}^T(s'_i)$ (where $\pi^{t}(s'_i)$ is the policy for the infostate in $G$ and $\bar{\pi}^T(s'_i)$ is the policy for the infostate in the subgame rooted at $\beta_{s'_{\text{pub}}}^{\pi^{t}}$). It is proven that as $T' \to \infty$, CFR-D converges to a $O(\frac{1}{\sqrt{T}})$-Nash equilibrium~\cite{burch2014solving}.

CFR-AVG is identical to CFR-D, except the subgames that are solved on each iteration $t$ are rooted at $\beta_{s'_{\text{pub}}}^{\bar{\pi}^{t}}$ rather than $\beta_{s'_{\text{pub}}}^{\pi^{t}}$. Theorem~\ref{th:cfrave} proves that CFR-AVG achieves the same bound on convergence to a Nash equilibrium as CFR-D.

\begin{theorem}
Suppose that $T$ iterations of CFR-AVG are run in a depth-limited subgame, where on each iteration $t \le T$ the subgame rooted at each leaf PBS $\beta_{s'_\text{pub}}^{\bar{\pi}^{t}}$ is solved completely. Then $\bar{\pi}^{T}$ is a $\frac{C}{\sqrt{T}}$-Nash equilibrium for a game-specific constant $C$.
\label{th:cfrave}
\end{theorem}

CFR-AVG has a number of potential benefits over CFR-D:
\begin{itemize}
    \item Since $\bar{\pi}^t$ converges to a Nash equilibrium as $t \rightarrow \infty$, CFR-AVG allows $\hat{v}$ to focus on being accurate over a more narrow subspace of inputs.
    
    \item When combined with a policy network (as introduced in Section~\ref{sec:method_policy}), CFR-AVG may allow $\hat{v}$ to focus on an even more narrow subspace of inputs.
    
    \item Since $\bar{\pi}^{t+1}$ is much closer to $\bar{\pi}^t$ than $\pi^{t+1}$ is to $\pi^t$, in practice as $t$ becomes large one can avoid querying the value network on every iteration and instead recycle the values from a previous iteration. This may be particularly valuable for Monte Carlo versions of CFR.
\end{itemize}

While CFR-AVG is theoretically sound, we modify its implementation in our experiments to make it more efficient in a way that has not been proven to be theoretically sound. The reason for this is that while the \emph{input} to the value network is $\beta_{s'_\text{pub}}^{\bar{\pi}^t}$ (i.e., the leaf PBS corresponding to $\bar{\pi}^t$ being played in $G$, the \emph{output} needs to be the value of each infostate $s_i$ given that $\pi^t$ is played in $G$. Thus, unlike CFR-D and FP, in CFR-AVG there is a mismatch between the input policy and the output policy.

One way to cope with this is to have the input consist of both $\beta_{s'_\text{pub}}^{\bar{\pi}^t}$ and $\beta_{s'_\text{pub}}^{\pi^t}$. However, we found this performed relatively poorly in preliminary experiments when trained through self play.

Instead, on iteration $t-1$ we store the output from $\hat{v}(s_i | \beta_{s'_\text{pub}}^{\bar{\pi}^{t-1}})$ for each $s_i$ and on iteration $t$ we set $v^t(s_i)$ to be $t \hat{v}(s_i | \beta_{s'_\text{pub}}^{\bar{\pi}^{t}}) - (t-1) \hat{v}(s_i | \beta_{s'_\text{pub}}^{\bar{\pi}^{t-1}})$ (in vanilla CFR). The motivation for this is that $\pi^t = t \bar{\pi}^t - (t - 1) \bar{\pi}^{t-1}$. If $v^t(h) = v^{t-1}(h)$ for each history $h$ in the leaf PBS, then this modification of CFR-AVG is sound. Since $v^t(h) = v^{t-1}(h)$ when $h$ is a full-game terminal node (i.e., it has no actions), this modified form of CFR-AVG is identical to CFR in a non-depth-limited game. However, that is not the case in a depth-limited subgame, and it remains an open question whether this modified form of CFR-AVG is theoretically sound in depth-limited subgames. Empirically, however, we found that it converges to a Nash equilibrium in turn endgame hold'em for every set of parameters (e.g., bet sizes, stack sizes, and initial beliefs) that we tested.

Figure~\ref{fig:cfrave} shows the performance of CFR-D, CFR-AVG, our modified form of CFR-AVG, and FP in TEH when using an oracle function for the value network. It also shows the performance of CFR-D, our modified form of CFR-AVG, and FP in TEH when using a value network trained through self-play. Surprisingly, the theoretically sound form of CFR-AVG does worse than CFR-D when using an oracle function. However, the modified form of CFR-AVG does better than CFR-D when using an oracle function and also when trained through self play.

\begin{figure}[!h]
	\vspace{-0.1in}
	\centering
	\includegraphics[width=0.95\textwidth]{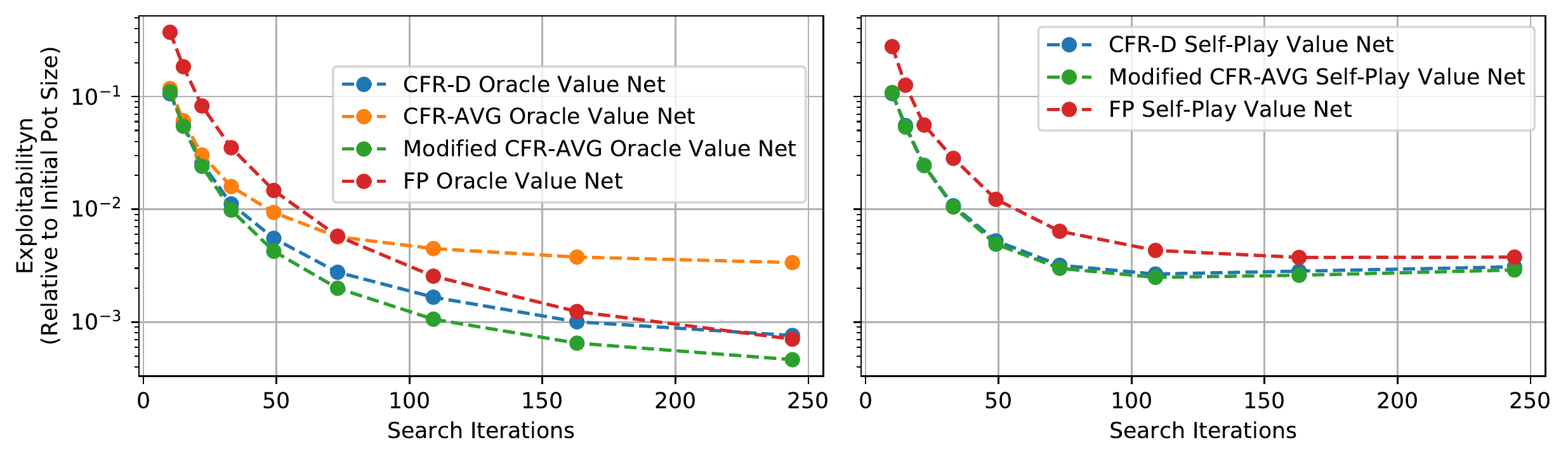}
	\vspace{-0.1in}

	\caption{\small{\textbf{Left:} comparison of CFR-D, CFR-AVG, modified CFR-AVG, and FP using an oracle value network which returns exact values for leaf PBSs. \textbf{Right:} comparison of CFR-D, modified CFR-AVG, and FP using a value network learned through 300 epochs of self play.
	}}
	\label{fig:cfrave}
\end{figure}

We also trained a model on HUNL with training parameters that were identical to the one reported in Section~\ref{sec:results}, but using CFR-D rather than CFR-AVG. That model lost to BabyTartanian8 by $10 \pm 3$ whereas the CFR-AVG model won by $9 \pm 4$. The CFR-D model also beat Slumbot by $39 \pm 6$ whereas the CFR-AVG model won by $45 \pm 5$.

\subsection{Proof of Theorem~\ref{th:cfrave}}

Our proof closely follows from~\cite{burch2014solving} and~\cite{moravvcik2017deepstack}.

\begin{proof}
Let $R^t(s_i)$ be the (cumulative) regret of infostates $s_i$ on iteration $t$. We show that the regrets of all infostates in $G'$ are bounded by $O(\sqrt{T})$ and therefore the regret of the entire game is bounded by $O(\sqrt{T})$.

First, consider the infostates in $G$. Since their policies are chosen according to CFR each iteration, their regrets are bounded by $O(\sqrt{T})$ regardless of the policies played in descendant infostates.

Next consider an infostate $s_i \in G' \setminus G$. We prove inductively that $R^t(s_i) \le 0$. Let $\beta^{\pi^t}$ be the PBS at the root of the subgame containing $s_i$ in CFR-D, and $\beta^{\bar{\pi}^t}$ be the PBS at the root of the subgame containing $s_i$ in CFR-AVG. On the first iteration, $\beta^{\pi^t} = \beta^{\bar{\pi}^t}$. Since we assume CFR-AVG computes an exact equilibrium in the subgame rooted at $\beta^{\bar{\pi}^t} = \beta^{\pi^t}$, so $R^t(s_i) = 0$ on the first iteration.

Next, we prove $R^{t+1}(s_i) \le R^{t}(s_i)$. We define $a^{*,t}$ as
\begin{equation}
a_i^{*,t} = \argmax_{a_i} \sum_{t'=0}^{t} v^{t'}(s_i,a_i)
\end{equation}

By definition of regret,
\begin{equation}
R^{t+1}(s_i) = \sum_{t'=0}^{t+1} \big( v^{t'}(s_i,a_i^{*,t+1}) - v^{t'}(s_i)\big)
\end{equation}
Separating iteration $t+1$ from the summation we get
\begin{equation}
R^{t+1}(s_i) = \sum_{t'=0}^t \big( v^{t'}(s_i,a_i^{*,t+1}) - v^{t'}(s_i) \big) + \big(v^{t+1}(s_i,a_i^{*,t+1}) - v^{t+1}(s_i)\big)
\end{equation}

By definition of $a_i^{*,t}$ we know $\sum_{t'=0}^t v^{t'}(s_i,a_i^{*,t+1}) \le \sum_{t'=0}^t v^{t'}(s_i,a_i^{*,t})$, so
\begin{equation}
R^{t+1}(s_i) \le \sum_{t'=0}^t \big( v^{t'}(s_i,a_i^{*,t}) - v^{t'}(s_i) \big) + \big(v^{t+1}(s_i,a_i^{*,t+1}) - v^{t+1}(s_i)\big)
\end{equation}
Since $\sum_{t'=0}^t \big( v^{t'}(s_i,a_i^{*,t}) - v^{t'}(s_i) \big)$ is the definition of $R^{t}(s_i)$ we get
\begin{equation}
R^{t+1}(s_i) \le R^{t}(s_i) + \big(v^{t+1}(s_i,a_i^{*,t+1}) - v^{t+1}(s_i)\big)
\end{equation}
Since $\pi^{t+1} = \pi^{*,t+1}$ in the subgame where $\pi^{*,t+1}$ is an exact equilibrium of the subgame rooted at $\beta^{\bar{\pi}^{t+1}}$, so $\pi^{t+1}$ is a best response to $\bar{\pi}^{t+1}$ in the subgame and therefore $v^{t+1}(s_i,a_i^{*,t+1}) = v^{t+1}(s_i)$. Thus,
\begin{equation}
    R^{t+1}(s_i) \le R^t(s_i)
\end{equation}
\end{proof}

%% file: core/A9_warm.tex
\section{CFR Warm Start Algorithm Used}
\label{sec:warm}

Our warm start technique for CFR is based on~\cite{brown2016strategy}, which requires only a policy profile to warm start CFR soundly. That techniques computes a ``soft'' best response to the policy profile, which results in instantaneous regrets for each infostate. Those instantaneous regrets are scaled up to be equivalent to some number of CFR iterations. However, that technique requires careful parameter tuning to achieve good performance in practice.

We instead use a simplified warm start technique in which an exact best response to the policy profile is computed. That best response results in instantaneous regrets at each infostate. Those regrets are scaled up by a factor of 15 to imitate 15 CFR iterations. Similarly, the average policy effectively assumes that the warm start policy was played for the first 15 iterations of CFR. CFR then proceeds as if 15 iterations have already occurred.

%% file: main.bbl
\begin{thebibliography}{10}

\bibitem{anthony2017thinking}
Thomas Anthony, Zheng Tian, and David Barber.
\newblock Thinking fast and slow with deep learning and tree search.
\newblock In {\em Advances in Neural Information Processing Systems}, pages
  5360--5370, 2017.

\bibitem{aumann1976agreeing}
Robert~J Aumann.
\newblock Agreeing to disagree.
\newblock {\em The annals of statistics}, pages 1236--1239, 1976.

\bibitem{ba2016layer}
Jimmy~Lei Ba, Jamie~Ryan Kiros, and Geoffrey~E Hinton.
\newblock Layer normalization.
\newblock {\em arXiv preprint arXiv:1607.06450}, 2016.

\bibitem{bowling2015heads}
Michael Bowling, Neil Burch, Michael Johanson, and Oskari Tammelin.
\newblock Heads-up limit hold’em poker is solved.
\newblock {\em Science}, 347(6218):145--149, 2015.

\bibitem{brown1951iterative}
George~W Brown.
\newblock Iterative solution of games by fictitious play.
\newblock {\em Activity analysis of production and allocation}, 13(1):374--376,
  1951.

\bibitem{brown2015hierarchical}
Noam Brown, Sam Ganzfried, and Tuomas Sandholm.
\newblock Hierarchical abstraction, distributed equilibrium computation, and
  post-processing, with application to a champion no-limit texas hold'em agent.
\newblock In {\em Proceedings of the 2015 International Conference on
  Autonomous Agents and Multiagent Systems}, pages 7--15. International
  Foundation for Autonomous Agents and Multiagent Systems, 2015.

\bibitem{brown2019deep}
Noam Brown, Adam Lerer, Sam Gross, and Tuomas Sandholm.
\newblock Deep counterfactual regret minimization.
\newblock In {\em International Conference on Machine Learning}, pages
  793--802, 2019.

\bibitem{brown2015simultaneous}
Noam Brown and Tuomas Sandholm.
\newblock Simultaneous abstraction and equilibrium finding in games.
\newblock In {\em Twenty-Fourth International Joint Conference on Artificial
  Intelligence}, 2015.

\bibitem{brown2016baby}
Noam Brown and Tuomas Sandholm.
\newblock Baby tartanian8: Winning agent from the 2016 annual computer poker
  competition.
\newblock In {\em IJCAI}, pages 4238--4239, 2016.

\bibitem{brown2016strategy}
Noam Brown and Tuomas Sandholm.
\newblock Strategy-based warm starting for regret minimization in games.
\newblock In {\em Thirtieth AAAI Conference on Artificial Intelligence}, 2016.

\bibitem{brown2017safe}
Noam Brown and Tuomas Sandholm.
\newblock Safe and nested subgame solving for imperfect-information games.
\newblock In {\em Advances in neural information processing systems}, pages
  689--699, 2017.

\bibitem{brown2017superhuman}
Noam Brown and Tuomas Sandholm.
\newblock Superhuman {A}{I} for heads-up no-limit poker: Libratus beats top
  professionals.
\newblock {\em Science}, page eaao1733, 2017.

\bibitem{brown2019solving}
Noam Brown and Tuomas Sandholm.
\newblock Solving imperfect-information games via discounted regret
  minimization.
\newblock In {\em Proceedings of the AAAI Conference on Artificial
  Intelligence}, volume~33, pages 1829--1836, 2019.

\bibitem{brown2019superhuman}
Noam Brown and Tuomas Sandholm.
\newblock Superhuman {A}{I} for multiplayer poker.
\newblock {\em Science}, page eaay2400, 2019.

\bibitem{brown2018depth}
Noam Brown, Tuomas Sandholm, and Brandon Amos.
\newblock Depth-limited solving for imperfect-information games.
\newblock In {\em Advances in Neural Information Processing Systems}, pages
  7663--7674, 2018.

\bibitem{burch2014solving}
Neil Burch, Michael Johanson, and Michael Bowling.
\newblock Solving imperfect information games using decomposition.
\newblock In {\em Twenty-Eighth AAAI Conference on Artificial Intelligence},
  2014.

\bibitem{burch2018aivat}
Neil Burch, Martin Schmid, Matej Moravcik, Dustin Morill, and Michael Bowling.
\newblock Aivat: A new variance reduction technique for agent evaluation in
  imperfect information games.
\newblock In {\em Thirty-Second AAAI Conference on Artificial Intelligence},
  2018.

\bibitem{chiang2012online}
Chao-Kai Chiang, Tianbao Yang, Chia-Jung Lee, Mehrdad Mahdavi, Chi-Jen Lu, Rong
  Jin, and Shenghuo Zhu.
\newblock Online optimization with gradual variations.
\newblock In {\em Conference on Learning Theory}, pages 6--1, 2012.

\bibitem{dibangoye2016optimally}
Jilles~Steeve Dibangoye, Christopher Amato, Olivier Buffet, and Fran{\c{c}}ois
  Charpillet.
\newblock Optimally solving dec-pomdps as continuous-state mdps.
\newblock {\em Journal of Artificial Intelligence Research}, 55:443--497, 2016.

\bibitem{foerster2019bayesian}
Jakob Foerster, Francis Song, Edward Hughes, Neil Burch, Iain Dunning, Shimon
  Whiteson, Matthew Botvinick, and Michael Bowling.
\newblock Bayesian action decoder for deep multi-agent reinforcement learning.
\newblock In {\em International Conference on Machine Learning}, pages
  1942--1951, 2019.

\bibitem{ganzfried2014potential}
Sam Ganzfried and Tuomas Sandholm.
\newblock Potential-aware imperfect-recall abstraction with earth mover's
  distance in imperfect-information games.
\newblock In {\em Proceedings of the Twenty-Eighth AAAI Conference on
  Artificial Intelligence}, pages 682--690, 2014.

\bibitem{ganzfried2015endgame}
Sam Ganzfried and Tuomas Sandholm.
\newblock Endgame solving in large imperfect-information games.
\newblock In {\em Proceedings of the 2015 International Conference on
  Autonomous Agents and Multiagent Systems}, pages 37--45. International
  Foundation for Autonomous Agents and Multiagent Systems, 2015.

\bibitem{gelly2007combining}
Sylvain Gelly and David Silver.
\newblock Combining online and offline knowledge in uct.
\newblock In {\em Proceedings of the 24th international conference on Machine
  learning}, pages 273--280, 2007.

\bibitem{gilpin2005optimal}
Andrew Gilpin and Tuomas Sandholm.
\newblock Optimal rhode island hold'em poker.
\newblock In {\em Proceedings of the 20th national conference on Artificial
  intelligence-Volume 4}, pages 1684--1685, 2005.

\bibitem{gilpin2006competitive}
Andrew Gilpin and Tuomas Sandholm.
\newblock A competitive texas hold'em poker player via automated abstraction
  and real-time equilibrium computation.
\newblock In {\em Proceedings of the National Conference on Artificial
  Intelligence}, volume~21, page 1007. Menlo Park, CA; Cambridge, MA; London;
  AAAI Press; MIT Press; 1999, 2006.

\bibitem{hansen2004dynamic}
Eric~A Hansen, Daniel~S Bernstein, and Shlomo Zilberstein.
\newblock Dynamic programming for partially observable stochastic games.
\newblock In {\em AAAI}, volume~4, pages 709--715, 2004.

\bibitem{hendrycks2016gaussian}
Dan Hendrycks and Kevin Gimpel.
\newblock Gaussian error linear units (gelus).
\newblock {\em arXiv preprint arXiv:1606.08415}, 2016.

\bibitem{hoda2010smoothing}
Samid Hoda, Andrew Gilpin, Javier Pena, and Tuomas Sandholm.
\newblock Smoothing techniques for computing nash equilibria of sequential
  games.
\newblock {\em Mathematics of Operations Research}, 35(2):494--512, 2010.

\bibitem{horak2019solving}
Karel Hor{\'a}k and Branislav Bo{\v{s}}ansk{\`y}.
\newblock Solving partially observable stochastic games with public
  observations.
\newblock In {\em Proceedings of the AAAI Conference on Artificial
  Intelligence}, volume~33, pages 2029--2036, 2019.

\bibitem{johanson2012finding}
Michael Johanson, Nolan Bard, Neil Burch, and Michael Bowling.
\newblock Finding optimal abstract strategies in extensive-form games.
\newblock In {\em Proceedings of the Twenty-Sixth AAAI Conference on Artificial
  Intelligence}, pages 1371--1379, 2012.

\bibitem{kaelbling1998planning}
Leslie~Pack Kaelbling, Michael~L Littman, and Anthony~R Cassandra.
\newblock Planning and acting in partially observable stochastic domains.
\newblock {\em Artificial intelligence}, 101(1-2):99--134, 1998.

\bibitem{kingma2014adam}
Diederik~P Kingma and Jimmy Ba.
\newblock Adam: A method for stochastic optimization.
\newblock {\em arXiv preprint arXiv:1412.6980}, 2014.

\bibitem{kovavrik2019problems}
Vojt{\v{e}}ch Kova{\v{r}}{\'\i}k and Viliam Lis{\`y}.
\newblock Problems with the efg formalism: a solution attempt using
  observations.
\newblock {\em arXiv preprint arXiv:1906.06291}, 2019.

\bibitem{kovavrik2019rethinking}
Vojt{\v{e}}ch Kova{\v{r}}{\'\i}k, Martin Schmid, Neil Burch, Michael Bowling,
  and Viliam Lis{\`y}.
\newblock Rethinking formal models of partially observable multiagent decision
  making.
\newblock {\em arXiv preprint arXiv:1906.11110}, 2019.

\bibitem{kroer2018solving}
Christian Kroer, Gabriele Farina, and Tuomas Sandholm.
\newblock Solving large sequential games with the excessive gap technique.
\newblock In {\em Advances in Neural Information Processing Systems}, pages
  864--874, 2018.

\bibitem{kroer2018faster}
Christian Kroer, Kevin Waugh, Fatma K{\i}l{\i}n{\c{c}}-Karzan, and Tuomas
  Sandholm.
\newblock Faster algorithms for extensive-form game solving via improved
  smoothing functions.
\newblock {\em Mathematical Programming}, pages 1--33, 2018.

\bibitem{lerer2020improving}
Adam Lerer, Hengyuan Hu, Jakob Foerster, and Noam Brown.
\newblock Improving policies via search in cooperative partially observable
  games.
\newblock In {\em AAAI Conference on Artificial Intelligence}, 2020.

\bibitem{leslie2006generalised}
David~S Leslie and Edmund~J Collins.
\newblock Generalised weakened fictitious play.
\newblock {\em Games and Economic Behavior}, 56(2):285--298, 2006.

\bibitem{lisy2017eqilibrium}
Viliam Lisy and Michael Bowling.
\newblock Eqilibrium approximation quality of current no-limit poker bots.
\newblock In {\em Workshops at the Thirty-First AAAI Conference on Artificial
  Intelligence}, 2017.

\bibitem{moravvcik2017deepstack}
Matej Morav{\v{c}}{\'\i}k, Martin Schmid, Neil Burch, Viliam Lis{\`y}, Dustin
  Morrill, Nolan Bard, Trevor Davis, Kevin Waugh, Michael Johanson, and Michael
  Bowling.
\newblock Deepstack: Expert-level artificial intelligence in heads-up no-limit
  poker.
\newblock {\em Science}, 356(6337):508--513, 2017.

\bibitem{moravcik2016refining}
Matej Moravcik, Martin Schmid, Karel Ha, Milan Hladik, and Stephen~J
  Gaukrodger.
\newblock Refining subgames in large imperfect information games.
\newblock In {\em Thirtieth AAAI Conference on Artificial Intelligence}, 2016.

\bibitem{nash1951non}
John Nash.
\newblock Non-cooperative games.
\newblock {\em Annals of mathematics}, pages 286--295, 1951.

\bibitem{nayyar2013decentralized}
Ashutosh Nayyar, Aditya Mahajan, and Demosthenis Teneketzis.
\newblock Decentralized stochastic control with partial history sharing: A
  common information approach.
\newblock {\em IEEE Transactions on Automatic Control}, 58(7):1644--1658, 2013.

\bibitem{newman2016reconnaissance}
Andrew~J Newman, Casey~L Richardson, Sean~M Kain, Paul~G Stankiewicz, Paul~R
  Guseman, Blake~A Schreurs, and Jeffrey~A Dunne.
\newblock Reconnaissance blind multi-chess: an experimentation platform for isr
  sensor fusion and resource management.
\newblock In {\em Signal Processing, Sensor/Information Fusion, and Target
  Recognition XXV}, volume 9842, page 984209. International Society for Optics
  and Photonics, 2016.

\bibitem{oliehoek2013sufficient}
Frans~Adriaan Oliehoek.
\newblock Sufficient plan-time statistics for decentralized pomdps.
\newblock In {\em Twenty-Third International Joint Conference on Artificial
  Intelligence}, 2013.

\bibitem{paszke2019pytorch}
Adam Paszke, Sam Gross, Francisco Massa, Adam Lerer, James Bradbury, Gregory
  Chanan, Trevor Killeen, Zeming Lin, Natalia Gimelshein, Luca Antiga, et~al.
\newblock Pytorch: An imperative style, high-performance deep learning library.
\newblock In {\em Advances in neural information processing systems}, pages
  8026--8037, 2019.

\bibitem{rakhlin2013online}
Alexander Rakhlin and Karthik Sridharan.
\newblock Online learning with predictable sequences.
\newblock 2013.

\bibitem{rockafellar1970convex}
R~Tyrrell Rockafellar.
\newblock {\em Convex analysis}.
\newblock Number~28. Princeton university press, 1970.

\bibitem{samuel1959some}
Arthur~L Samuel.
\newblock Some studies in machine learning using the game of checkers.
\newblock {\em IBM Journal of research and development}, 3(3):210--229, 1959.

\bibitem{schrittwieser2019mastering}
Julian Schrittwieser, Ioannis Antonoglou, Thomas Hubert, Karen Simonyan,
  Laurent Sifre, Simon Schmitt, Arthur Guez, Edward Lockhart, Demis Hassabis,
  Thore Graepel, et~al.
\newblock Mastering atari, go, chess and shogi by planning with a learned
  model.
\newblock {\em arXiv preprint arXiv:1911.08265}, 2019.

\bibitem{schulman2015high}
John Schulman, Philipp Moritz, Sergey Levine, Michael~I. Jordan, and Pieter
  Abbeel.
\newblock High-dimensional continuous control using generalized advantage
  estimation.
\newblock In {\em International Conference on Learning Representations (ICLR)},
  2016.

\bibitem{seitz2019value}
Dominik Seitz, Vojtech Kovar{\'\i}k, Viliam Lis{\`y}, Jan Rudolf, Shuo Sun, and
  Karel Ha.
\newblock Value functions for depth-limited solving in imperfect-information
  games beyond poker.
\newblock {\em arXiv preprint arXiv:1906.06412}, 2019.

\bibitem{serrino2019finding}
Jack Serrino, Max Kleiman-Weiner, David~C Parkes, and Josh Tenenbaum.
\newblock Finding friend and foe in multi-agent games.
\newblock In {\em Advances in Neural Information Processing Systems}, pages
  1249--1259, 2019.

\bibitem{shannon1950programming}
Claude~E Shannon.
\newblock Programming a computer for playing chess.
\newblock {\em The London, Edinburgh, and Dublin Philosophical Magazine and
  Journal of Science}, 41(314):256--275, 1950.

\bibitem{silver2018general}
David Silver, Thomas Hubert, Julian Schrittwieser, Ioannis Antonoglou, Matthew
  Lai, Arthur Guez, Marc Lanctot, Laurent Sifre, Dharshan Kumaran, Thore
  Graepel, et~al.
\newblock A general reinforcement learning algorithm that masters chess, shogi,
  and go through self-play.
\newblock {\em Science}, 362(6419):1140--1144, 2018.

\bibitem{silver2017mastering}
David Silver, Julian Schrittwieser, Karen Simonyan, Ioannis Antonoglou, Aja
  Huang, Arthur Guez, Thomas Hubert, Lucas Baker, Matthew Lai, Adrian Bolton,
  et~al.
\newblock Mastering the game of go without human knowledge.
\newblock {\em Nature}, 550(7676):354, 2017.

\bibitem{sustr2019monte}
Michal {\v{S}}ustr, Vojt{\v{e}}ch Kova{\v{r}}{\'\i}k, and Viliam Lis{\`y}.
\newblock Monte carlo continual resolving for online strategy computation in
  imperfect information games.
\newblock In {\em Proceedings of the 18th International Conference on
  Autonomous Agents and MultiAgent Systems}, pages 224--232. International
  Foundation for Autonomous Agents and Multiagent Systems, 2019.

\bibitem{syrgkanis2015fast}
Vasilis Syrgkanis, Alekh Agarwal, Haipeng Luo, and Robert~E Schapire.
\newblock Fast convergence of regularized learning in games.
\newblock In {\em Advances in Neural Information Processing Systems}, pages
  2989--2997, 2015.

\bibitem{tesauro1994td}
Gerald Tesauro.
\newblock T{D}-{G}ammon, a self-teaching backgammon program, achieves
  master-level play.
\newblock {\em Neural computation}, 6(2):215--219, 1994.

\bibitem{van2000weakened}
Ben Van~der Genugten.
\newblock A weakened form of fictitious play in two-person zero-sum games.
\newblock {\em International Game Theory Review}, 2(04):307--328, 2000.

\bibitem{zinkevich2008regret}
Martin Zinkevich, Michael Johanson, Michael Bowling, and Carmelo Piccione.
\newblock Regret minimization in games with incomplete information.
\newblock In {\em Advances in neural information processing systems}, pages
  1729--1736, 2008.

\end{thebibliography}
